\documentclass[11pt]{article}
\usepackage[a4paper,hmargin={2.5cm,2.5cm},vmargin={1.0cm,2.5cm},includehead]{geometry}

\usepackage[utf8]{inputenc} 
\usepackage[T1]{fontenc}    
\usepackage[bookmarks, colorlinks=true, plainpages = false, colorlinks=true,
            citecolor=red,
            linkcolor=blue,
            anchorcolor=red,
            urlcolor=blue]{hyperref}
\usepackage{url}

\usepackage{nicefrac}       
\usepackage{microtype}      
\usepackage{xcolor}         

\usepackage{amsmath}
\usepackage{amssymb}
\usepackage{amsfonts}
\usepackage{amsbsy}
\usepackage{bm}
\usepackage{appendix}
\usepackage{thmtools}
\usepackage{thm-restate}
\usepackage{mathtools}
\usepackage{bbm}
\usepackage{cleveref}
\usepackage{xfrac}

\usepackage{subfigure}
\usepackage{float}
\usepackage{dsfont}
\usepackage{graphicx}
\usepackage{booktabs}
\usepackage{color}
\usepackage{url}

\usepackage{pdfpages}
\usepackage{enumitem}

\usepackage[linesnumbered,ruled]{algorithm2e}
\newcommand{\mycomment}[1]{\tcp{\texttt{ #1}}}

\usepackage{parskip}
\usepackage{natbib}
\newtheorem{theorem}{Theorem}
\newtheorem{corollary}[theorem]{Corollary}
\newtheorem{lemma}[theorem]{Lemma}

\newtheorem{remark}[theorem]{Remark}
\newtheorem{definition}[theorem]{Definition}

\newcommand{\1}{\mathbbm{1}} 
\newcommand{\R}{\mathbb{R}} 
\renewcommand{\P}[1]{\mathbb{P}\left(#1 \right)} 
\newcommand{\E}[1]{\mathbb{E}\left[ #1 \right]} 

\newcommand{\eps}{\varepsilon} 

\newcommand{\aln}[1]{\begin{align}  #1 \end{align}}
\newcommand{\alns}[1]{\begin{align*}  #1 \end{align*}}

\newcommand{\pbr}[1]{\left( #1 \right)} 
\newcommand{\sbr}[1]{\left[ #1\right]}
\newcommand{\cbr}[1]{\left\{ #1\right\}}

\newcommand{\Bin}{\text{Bin}}
\newcommand{\Ht}{\widetilde{\mathcal{H}}}
\newcommand{\ct}{\widetilde{c}}

\newcommand{\bfP}{\mathbf{P}}
\newcommand{\bfQ}{\mathbf{Q}}

\newcommand{\indicator}{\mathbbm{1}}

\newcommand{\pihat}{\widehat{\pi}}
\newcommand{\pihatMLE}{\widehat{\pi}^{\text{MLE}}}

\newcommand{\pistar}{\pi^*}

\newcommand{\muhat}{\widehat{\mu}}
\newcommand{\nuhat}{\widehat{\nu}}

\newcommand{\tmu}{\widetilde{\mu}}

\newcommand{\tN}{\widetilde{N}}

\newcommand{\tZ}{\widetilde{Z}}

\newcommand{\Bern}{\mathrm{Bern}}

\renewcommand{\deg}{\delta}
\newcommand{\core}{\mathrm{core}}
\newcommand{\dom}{\mathrm{dom}}

\newcommand{\newp}{q}
\newcommand{\newd}{r}
\newcommand{\newD}{R}

\newcommand{\calH}{\mathcal{H}}

\newcommand{\IHG}{\mathrm{RG}}
\renewcommand{\indicator}[1]{\bm{1}_{\{ #1 \}}}

\definecolor{warningcol}{rgb}{.99,.1,.5}
\definecolor{todocol}{rgb}{.4,.4,.8}
\definecolor{sketchcol}{rgb}{.4,.4,.8}
\definecolor{outlinecol}{rgb}{.8,.4,.3}

\DeclareMathOperator*{\argmax}{arg\,max}
\DeclareMathOperator*{\argmin}{arg\,min}

\begin{document}

\title{Aligning Multiple Inhomogeneous Random Graphs: Fundamental Limits of Exact Recovery}
\author{Taha Ameen and Bruce Hajek 
}

\date{\today}

\markboth{Aligning Multiple Inhomogeneous Random Graphs: Fundamental Limits of Exact Recovery}{}

\maketitle

\begin{abstract}
    This work studies fundamental limits for recovering the underlying correspondence among \textit{multiple} correlated graphs. In the setting of inhomogeneous random graphs, we present and analyze a matching algorithm: first partially match the graphs pairwise and then combine the partial matchings by transitivity. Our analysis yields a sufficient condition on the problem parameters to exactly match all nodes across all the graphs. In the setting of homogeneous (Erd\H{o}s-R\'{e}nyi) graphs, we show that this condition is also necessary, i.e. the algorithm works down to the information theoretic threshold. This reveals a scenario where exact matching between two graphs alone is impossible, but leveraging more than two graphs allows exact matching among all the graphs. Converse results are also given in the inhomogeneous setting and transitivity again plays a role. Along the way, we derive independent results about the $k$-core of inhomogeneous random graphs.
\end{abstract}

\tableofcontents

\section{Introduction}\label{sec-introduction}
    
The information age has ushered an abundance of correlated networked data. For instance, the network structure of two social networks such as Facebook and Twitter is correlated because users are likely to connect with the same individuals in both networks. This wealth of correlated data presents both opportunities and challenges. On one hand, information from various datasets can be combined to increase the fidelity of data, which translates to better performance in downstream tasks. On the other hand, the interconnected nature of this data also raises privacy and security concerns. Linkage attacks, for instance, exploit correlated data to identify individuals in an anonymized network by linking to other sources~\citep{narayanan2009deanonymizing}.

The problem of graph alignment is to recover the underlying latent node correspondence between correlated networks, i.e. to match the nodes across graphs. The problem derives its importance from the fact that nodes in networks are often \textit{unlabeled} or \textit{scrambled} -- for example, user identities in social networks are anonymized for privacy. Applications include:

\begin{enumerate}
    \item Social networks: Matching two social networks amounts to recovering the hidden node identities of users in the anonymized network. Narayanan and Shmatikov~\citep{narayanan2009deanonymizing} established the efficacy of this approach: A third of the users active on both Twitter and Flickr could be identified with only a 12\% error rate, simply by matching the anonymous Twitter graph to the Flickr graph. In a related work~\citep{narayanan2008robust}, they also showed that anonymous movie ratings of Netflix users can be recovered by matching to the public Internet Movie Database (IMDb).
    \item Biological networks: The protein-protein interactomes of two different species are often correlated networks. Singh and co-authors~\citep{singh2008global} showed that matching the PPI network of yeast to that of fruit flies allowed identification of conserved functional components between the two species. This transfer of knowledge between species is a useful tool in biology and genomics~\citep{kazemi2016proper, bandyopadhyay2006systematic}. Similarly, the brain connectomes of healthy humans are correlated~\citep{sporns2005human}, and their alignment is useful in detecting abnormalities~\citep{calissano2024graph}.
    \item Natural language processing: Given a passage of text and a hypothesis, consider the objective of determining whether the hypothesis is entailed by the passage. Haghighi and co-authors~\citep{haghighi2005robust} presented an approach to solving this, by first constructing dependency trees of the text and hypothesis respectively, and then matching them to conclude with a yes/no answer after thresholding a score of the output matching. 
    \item Computer vision: Graph matching is employed extensively for vision tasks. One example is object detection: Schellewald and Schn\"{o}rr~\citep{schellewald2005probabilistic} showed that objects in images can be detected by matching graph representations of the test image to a feature image that contains the object. Yet another example in pattern recognition is the problem of tracking the movement of objects in a sequence of frames~\citep{yan2016short}.
    \item Autonomous agents: Robots and self-driving cars rely on a multitude of sensors to collect data about their environment. Combining information from various sensors increases data fidelity. For example, Ye and co-authors~\citep{ye2022multi} showed that graph alignment is useful for scene-map construction when multiple robots explore an unknown environment.
\end{enumerate}

There is a long history of utilizing domain knowledge to implement graph alignment algorithms for these tasks -- see for example~\citep{yan2016short, vento2013graph}. Even so, the theoretical study of graph matching algorithms and their performance guarantees is a relatively recent undertaking. The vast majority of work has focused on pairs of random graphs, particularly Erd{\H{o}}s-R{\'e}nyi graphs. Presently, we study exact recovery thresholds with multiple inhomogeneous random graphs.

\subsection{Related Work}
The theoretical study of graph matching algorithms and their performance guarantees has focused on the homogeneous setting of Erd\H{o}s-R\'{e}nyi (ER) graphs. Pedarsani and Grossglauser~\citep{pedarsani2011privacy} introduced the subsampling model to generate two such correlated graphs: the model entails twice subsampling each edge independently from a parent ER graph to obtain two sibling graphs, both of which are marginally ER graphs themselves. The goal is then to match nodes between the two graphs to recover the underlying latent correspondence. This has been the framework of choice for many works that study graph matching. For example, Cullina and Kiyavash~\citep{cullina2016improved,cullina2017exact} studied the problem of \textit{exactly matching} two ER graphs, where the objective is to match \textit{all} vertices correctly. They identified a threshold phenomenon for this task: exact recovery is possible if the problem parameters are above a threshold, and impossible otherwise. Subsequently, threshold phenomena were also identified for \textit{partial} graph matching between ER graphs, where the objective is to match only a positive fraction of nodes~\citep{ganassali2021impossibility, hall2023partial, wu2022settling, ding2023densesubgraph, du2025optimal}. The case of almost-exact recovery -- where the objective is to match all but a negligible fraction of nodes -- was studied by Cullina and co-authors~\citep{cullina2019kcore}: a necessary condition for almost exact recovery was identified, and the same condition was shown to also sufficient be for the \textit{$k$-core estimator}; the estimator is described  in Section~\ref{sec-results}. This estimator proved useful for graph matching in other contexts such as the stochastic block model~\citep{gaudio2022exact} and inhomogeneous random graphs~\citep{racz2023matching}. Ameen and Hajek~\citep{ameen2023robust} showed some robustness properties of the $k$-core estimator in the context of matching ER graphs under node corruption. The estimator also plays an important role in the present work.

Other works have studied \textit{computational} aspects of graph alignment: Various \textit{efficient} algorithms have been proposed, including algorithms based on the spectrum of the graph adjacency matrices~\citep{fan2022spectral}, node degree and neighborhood based algorithms~\citep{dai2019canonicallabeling, ding2021degreeprofile, mao2023constant} as well as algorithms based on iterative methods~\citep{ding2023polynomial} and counting subgraphs~\citep{mao2023chandelier, barak2019nearly, chai2024efficient}. While all these works study Erd\H{o}s-R\'{e}nyi graphs, there is recent interest in computational aspects of graph alignment in other models such as the stochastic block model~\citep{chai2024efficient}, as well as inhomogeneous random graphs~\citep{ding2023efficiently}. Still other works have studied the graph matching problem under side information -- this includes settings where a subset of nodes are matched a priori and provided as seeds~\citep{pedarsani2011privacy, mossel2020seeded, lubars2018correcting}, the setting where node attributes are present~\citep{wang2024feasible, zhang2024attributed,yang2024exact}, and the setting where correlation is localized to a subset of nodes~\citep{huang2024information}.

Incorporating information from \textit{multiple} graphs to match them has been recognized as an important research direction, for instance in the work of Gaudio and co-authors~\citep{gaudio2022exact}. For example, a user is active, on average, on 6.7 social networks each month~\citep{GWI}. Similarly, reconciling protein-protein interaction networks among \textit{multiple} species is an important problem in computational biology~\citep{singh2008global}. These applications also necessitate going beyond ER graphs, since many real-world networks lack the homogeneity of the ER model. To our knowledge, a very limited number of papers consider matchings among multiple graphs: The works~\citep{josephs2021recovery} and~\citep{racz2021correlated} have different objectives, and note that it is possible to exactly match $m$ graphs whenever it is possible to exactly match any two graphs by pairwise matching all the graphs exactly. In contrast, we show that under appropriate conditions, it is possible to exactly match $m$ graphs even when no two graphs can be pairwise matched exactly. The work of Rácz and Zhang~\citep{racz2024harnessing} studies the interplay between graph matching and community detection for multiple correlated stochastic block models. The recent work of Vassaux and Massouli\'{e}~\citep{vassaux2025feasibility} studies the information-theoretic limit for exact recovery in the \textit{Gaussian} model, where the observations are multiple correlated Wigner matrices, and also the limit for partial recovery in the homogeneous Erd\H{o}s-R\'{e}nyi setting. Closest to our work is the work of R\'{a}cz and Sridhar~\citep{racz2023matching}, where sufficeint conditions for exact and partial graph matching between \textit{two} inhomoegenous random graphs are studied.

\subsection{Contributions}
The present work studies fundamental limits of exact recovery with multiple inhomogeneous random graphs. These graphs allow for heterogeneity by having each edge $\{i,j\}$ in the graph appear independently with probability $p_{ij}$, and subsume other models such as ER graphs and stochastic block models, and are closely related to random geometric graphs. 

We consider the generalization of the matching problem to $m$ graphs and propose an algorithm to combine information among the graphs. A key idea is that of \textit{transitive closure} -- our algorithm works by pairwise matching graphs partially and then boosting the matchings by transitivity. Our analysis yields a sufficient condition for exact recovery of multiple inhomogeneous graphs. In the homogeneous case where $p_{ij} = p$, we show that this condition is also necessary, thereby characterizing the sharp reconstruction threshold for exactly matching multiple Erd\H{o}s-R\'{e}nyi graphs. In the heterogeneous setting, we also provide converse results and transitive closure plays a role here as well. We present simulation results to corroborate the utility of transitive closure as a black-box, efficient and optimal bridge between the binary and $m$-ary matching problems.

\section{Model} \label{sec-model}
    Let $\bfP  = (p_{ij})_{i,j\in [n]}$ with $0\leq p_{ij}=p_{ji} \leq 1$ denote a symmetric matrix with zero diagonals. A graph $G$ is an inhomogeneous random graph on the vertex set $[n] \triangleq \{1,2,\cdots,n\}$ with parameter matrix $\bfP$, denoted by $G\sim \IHG(n,\bfP )$, if each edge $\{i,j\}$ with $i<j$ exists independently with probability $p_{ij}$. Consider the natural extension of the subsampling model to $m$ graphs. Let $s \in (0,1]$, and obtain $m$ correlated inhomogeneous graphs $G_1',\cdots,G_m'$ by subsampling each edge of a parent graph $G\sim\IHG(n,\bfP )$ independently with probability $s$. These sibling graphs are themselves inhomogeneous random graphs with parameter matrix $\bfP s$. Let $\pistar_{12},\cdots,\pistar_{1m}$ denote independent permutations that are each drawn uniformly at random from the set of all permutations on $[n]$. Obtain the graphs $G_2,\cdots,G_m$ by permuting the nodes of $G_k'$ according to $\pistar_{1k}$, for each $k=2,\cdots,m$. For example, $\{i,j\}$ is an edge in $G_k'$ if and only if $\{\pistar_{1k}(i),\pistar_{1k}(j)\}$ is an edge in $G_k$. An illustration of the subsampling model is provided in Figure~\ref{fig: model}.

\begin{figure}[t]
    \centering
    \includegraphics[width=0.99\linewidth]{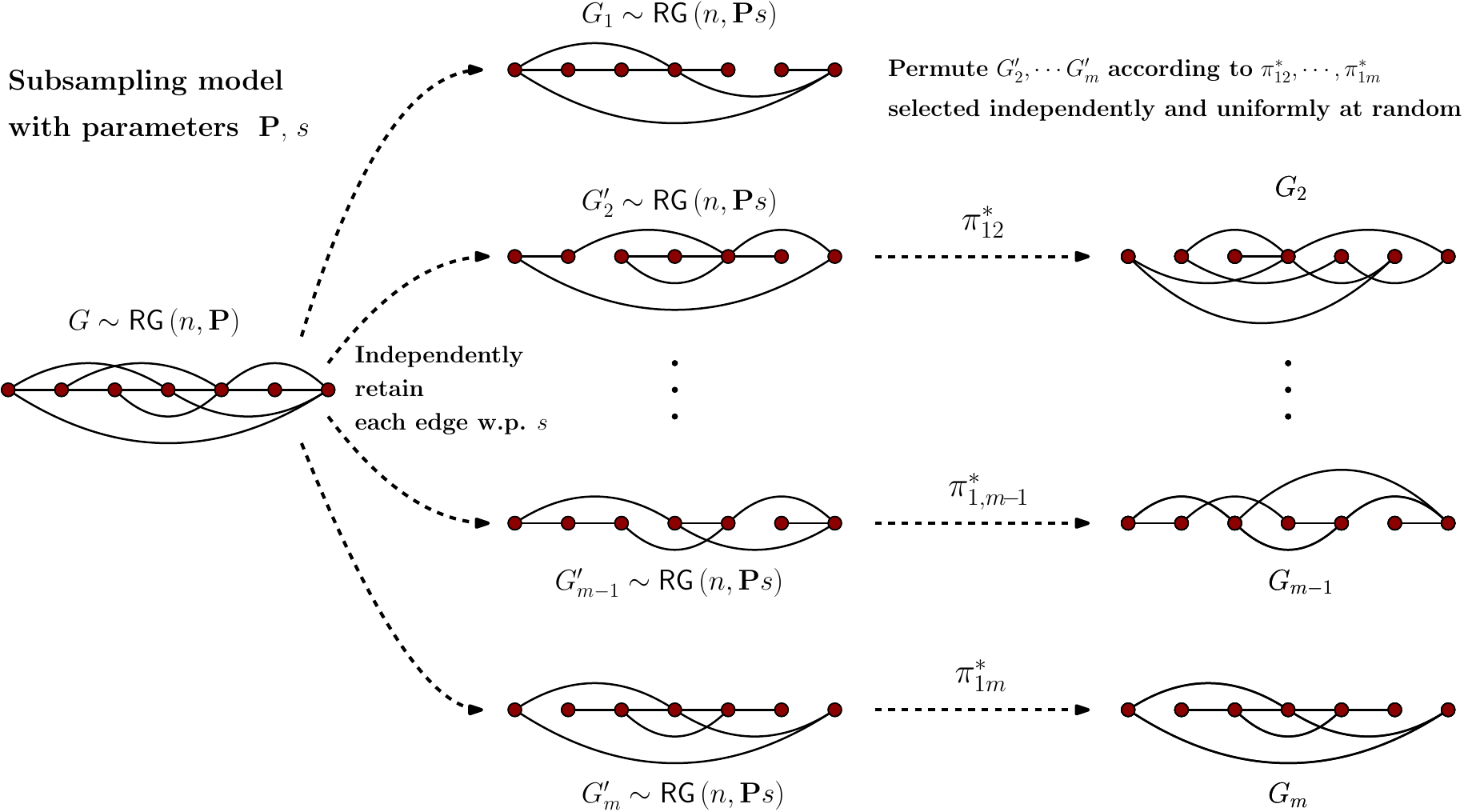}
    \caption{Obtaining $m$ correlated inhomogeneous graphs via the subsampling model.}
    \label{fig: model}
\end{figure}

We will use the tuple $(\pi_{12},\cdots,\pi_{1m})$ of permutations to identify the correspondence between all $m$ graphs. Stated thus, the graph $G_1$ is used as a reference and so we will interchangeably use $G_1$ and $G_1'$. Note that the underlying correspondence between all the graphs is fixed upon fixing $\pistar_{12},\cdots,\pistar_{1m}$: for any two graphs $G_k$ and $G_{\ell}$, their correspondence is given by $\pistar_{k\ell} \triangleq \pistar_{1\ell}  \circ (\pistar_{1k})^{-1}$.

In this work, we are concerned with \textit{exact recovery}. The goal is to determine necessary and sufficient conditions on the parameters $\bfP$, $s$ and $m$ so that given correlated inhomogeneous graphs $G_1,\cdots,G_m$ from the subsampling model, it is possible to exactly recover the underlying correspondences $\pistar_{12},\cdots,\pistar_{1m}$ with probability $1-o(1)$.

Formally, a \textit{matching} $(\mu_{12},\cdots,\mu_{1m})$ is a collection of injective functions with domain $\dom(\mu_{1i}) \subseteq V$ for each $i$, and co-domain $V$. An \textit{estimator} is simply a mechanism to map any collection of graphs $(G_1,\cdots, G_m)$ to a matching. We say that an estimator \textit{completely} matches the graphs if the output mappings $\mu_{12},\cdots \mu_{1m}$ are all complete, i.e. they are all permutations on $[n]$. If the matching $(\mu_{12},\cdots,\mu_{1m})$ satisfies $\mu_{1j} = \pistar_{1j}$ for all $2 \leq j \leq m$, then all nodes are correctly matched across all graphs, i.e. the graphs are exactly matched.

\paragraph{Some special cases of inhomogeneous random graphs}

The inhomogeneous random graph model captures a variety of generative models. Following~\citep{racz2023matching}, we introduce these models conditioned on latent variables intrinsic to them for convenience.
\begin{enumerate}
    \item Erd\H{o}s-R\'{e}nyi graph: Given a parameter $p \in [0,1]$, the homogeneous case of ER graphs sets $p_{ij} = p$ for all distinct $i$ and $j$.
    \item Stochastic block model (SBM)~\citep{abbe2018community, holland1983stochastic}: Given a partition $V_1,\cdots,V_M$ of the node set $[n]$ into $M$ communities and edge probabilities $(q_{ab})_{1\leq a < b \leq M}$, the SBM sets $p_{ij} = q_{ab}$ for all distinct $i$ and $j$ such that $i \in V_a$ and $j \in V_b$.
    \item Random geometric graph (RGG)~\citep{liu2023phase}: Given a collection $x_1,\cdots, x_n$ of points in a set $S \subseteq \mathbb{R}^d$, and parameters $p \in [0,1]$ and $r > 0$, the RGG sets $p_{ij} = p \cdot \mathbf{1}_{\cbr{\| x_i - x_j\|_2 \leq r}}$ for distinct $i$ and $j$.
    \item Chung-Lu graph (CLG)~\citep{chung2002average}: Given a collection $(w_i)_{1\leq i \leq n}$ of positive weights satisfying $w_i \leq \sqrt{\sum_{j\in [n]}w_j}$ for all $i$, the CLG model sets $p_{ij} = \frac{w_i w_j}{\sum_{k\in[n]}w_k}$ for distinct $i$ and $j$.
\end{enumerate}
Correlated graphs from any of the above models are obtained via the subsampling model. 

\paragraph{Notation} Throughout this work, $p_{\max}$ denotes the largest entry in the matrix $\bfP$, i.e. $p_{\max} = \max_{i,j}p_{ij}$. For a graph $G$ and node $v$ in $G$, let $\delta_G(v)$ denote the degree of $v$. Given a matching $\muhat$, we denote by $G^{\mu}$ the graph on node set given by the image of $\mu$ obtained by relabeling the nodes of $G$ according to $\mu$, i.e. for $i,j \in \dom(\mu)$:
\[
\{i,j\} \in E(G) \iff \{\mu(i),\mu(j)\} \in E(G^\mu).
\]
For any two graphs $H_1$ and $H_2$ on the same vertex set $V$, denote by $H_1 \vee H_2$ their \textit{union graph} and by $H_1\wedge H_2$ their \textit{intersection graph}. An edge $\{i,j\}$ is present in $H_1\vee H_2$ if it is present in either $H_1$ or $H_2$. Similarly, the edge is present in $H_1\wedge H_2$ if it is present in both $H_1$ and $H_2$. For a random variable $X$, denote by $M_X(t)\triangleq \mathbb{E}[e^{tX}]$ the moment generating function of $X$. Standard asymptotic notation $(O(\cdot),o(\cdot),\cdots)$ is used throughout, and it is implicit that $n\to\infty$.

\section{Main Results and Algorithms} \label{sec-results}
    
\subsection{Exact recovery threshold: Positive results}

Our first result is a sufficient condition for exactly matching all nodes across all the graphs.

\begin{theorem} \label{thm: achievability-m-graphs}
    Let $d_i = \sum_{j \neq i}p_{ij}$. Suppose that there exist $\alpha\in (0,1]$ and $\eps > 0$ such that $p_{\max} = o(n^{\alpha/2-\eps-1})$ and
    \begin{align} \label{eq: suff-condition-m-graphs}
        \sum_{i=1}^n e^{-d_i s(1-(1-s)^{m-1})}= o(n^{-\alpha}).
    \end{align}
    There exists an estimator, namely Algorithm~\ref{Alg: Match} described in the next subsection, whose output $\pihat_{12},\cdots,\pihat_{1m}$ satisfies
    \[ 
    \P{\pihat_{12} = \pistar_{12},\cdots,\pihat_{1m} = \pistar_{1m}} = 1-o(1).
    \]
\end{theorem}
    
A few remarks are in order. 
\begin{itemize}
\item The sufficient condition~\eqref{eq: suff-condition-m-graphs} of Theorem~\ref{thm: achievability-m-graphs} can be slightly improved to the condition 
    \begin{align}
        \sum_{i=1}^{n} e^{-d_i s(1-(1-s)^{m-1})} - \max_{j\in [n]} e^{-d_j s(1-(1-s)^{m-1})} = o(n^{-\alpha}). \label{eq: stronger-suff-condition-m-graphs}    
    \end{align}
    Let $v^* \in \argmin_{i \in [n]} d_i$. It can be shown that Algorithm~\ref{Alg: Match} applied to the graphs $G_k$ yields that all nodes except $v^*$ can be exactly matched. The node $v^*$ is then the only unmatched node and can be manually matched without ambiguity.
    \item R\'{a}cz and Sridhar~\citep{racz2023matching} studied the graph alignment problem when $m=2$. Setting $m=2$ in~\eqref{eq: suff-condition-m-graphs} yields $\sum_{i=1}^n e^{-d_i s^2} = o(n^{-\alpha})$, which is equivalent to their sufficient condition.
    \item Algorithm~\ref{Alg: Match} does not require access to the parameter matrix $\bfP$.
\end{itemize}

The sufficient condition~\eqref{eq: suff-condition-m-graphs} can be specialized to various models that are captured by inhomogeneous random graphs. Some of these are presented next.

\begin{corollary}[Correlated SBMs] Let $G_1,\cdots,G_m$ be correlated SBMs obtained from the subsampling model. Suppose there exist $\alpha\in(0,1]$ and $\eps>0$ such that $q_{ab} = o(n^{\alpha/2-\eps-1})$ for all $a$ and $b$. Suppose that 
\[ 
\sum_{a=1}^M |V_a| \cdot \exp\pbr{- s\pbr{1-(1-s)^{m-1}} \cdot \sum_{b=1}^m |V_b| q_{ab}} = o(n^{-\alpha}).
\]
There is an estimator whose output $\pihat_{12},\cdots,\pihat_{1m}$ satisfies $\P{\pihat_{12} = \pistar_{12},\cdots,\pihat_{1m} = \pistar_{1m}} = 1-o(1)$.    
\end{corollary}

\begin{corollary}[Correlated RGGs] Let $G_1,\cdots,G_m$ be correlated RGGs obtained from the subsampling model. Suppose there exist $\alpha\in(0,1]$ and $\eps>0$ such that $p = o(n^{\alpha/2-\eps-1})$. For a node $i$, let $N_r(i)$ denote the number of nodes in its $r$-neighborhood, i.e. $N_r(i) \triangleq |\{j: \|x_i - x_j\|_2 \leq r\}|$. Suppose that 
\[ 
\sum_{i=1}^{n} e^{-p s(1-(1-s)^{m-1}) \cdot N_r(i)} = o(n^{-\alpha}).
\]
There is an estimator whose output $\pihat_{12},\cdots,\pihat_{1m}$ satisfies $\P{\pihat_{12} = \pistar_{12},\cdots,\pihat_{1m} = \pistar_{1m}} = 1-o(1)$.
\end{corollary}

\begin{corollary}[Correlated CLGs] Let $G_1,\cdots,G_m$ be correlated CLGs obtained from the subsampling model. Suppose there exist $\alpha\in(0,1]$ and $\eps>0$ such that $w_i = o(n^{\alpha/2-\eps-1}) \sqrt{\sum_{k \in [n]}w_k}$ for all $i$. Suppose that 
\[ 
\sum_{i=1}^{n}e^{-w_i s(1-(1-s)^{m-1})}= o(n^{-\alpha}).
\]
There is an estimator whose output $\pihat_{12},\cdots,\pihat_{1m}$ satisfies $\P{\pihat_{12} = \pistar_{12},\cdots,\pihat_{1m} = \pistar_{1m}} = 1-o(1)$.
\end{corollary}

Next, we show that the sufficient condition~\eqref{eq: stronger-suff-condition-m-graphs} can be improved for a class of parameter matrices if the algorithm has access to $P$.  

\begin{theorem} \label{thm:informative_P}
Suppose $\eps$ and $s$ are fixed with $\eps >0$ and $0< s \leq 1.$  If $d$ varies with $n$ such that:
\begin{align}
    ds \geq (1+\epsilon)\log n, \label{eq:ds_bound}
\end{align}
then there exists a parameter matrix $P$ with $d_i \triangleq \sum_{j\neq i}P_{ij} \leq d$ for $i\in [n]$ and $p_{\max} \triangleq \max_{i,j} P_{i,j} = \frac{\log^2(n)}{n}$ such that exact recovery is possible, i.e. $m$ graphs $G_1, G_2,\cdots, G_m$ obtained by subsampling each edge independently with probability $s$ from a parent graph $G \sim \IHG(n,P)$ can be exactly matched with probability $1-o(1)$ by some algorithm that has access to $P$.
\end{theorem}

Taking $ds=(1+\epsilon)\log n$ in~\Cref{thm:informative_P} yields the following corollary.
\begin{corollary}  \label{cor: informative_P}
There exists a parameter matrix $P$ satisfying
    \begin{align}
        \sum_{i=1}^n e^{-d_i s(1-2\eps)} - \max_{j\in [n]}e^{-d_j s(1-2\eps)} = \Omega(n^{\eps}),
    \end{align}
    such that exact recovery is possible by some algorithm that has access to $P$.
\end{corollary}

\begin{remark}
    \Cref{cor: informative_P} shows that if $(1-s)^{m-1} > 2\eps$, then the sufficient condition~\eqref{eq: stronger-suff-condition-m-graphs} of~\Cref{thm: achievability-m-graphs} does not hold for some $P$, but exact recovery is nevertheless possible. Roughly speaking, the matrix $P$ itself is similar to an additional graph that can help in matching by some algorithm that has access to $P$.
\end{remark}

\subsection{Exact recovery threshold: negative results}
While~\Cref{cor: informative_P} shows that~\eqref{eq: stronger-suff-condition-m-graphs} is not a sharp sufficient condition for general parameter matrices, we prove that this condition is sharp in the homogeneous case of Erd\H{o}s-R\'{e}nyi graphs. This provides the sharp threshold for exactly matching $m$ ER graphs.

\begin{theorem}  \label{thm: impossibility-m-graphs}
    Let $C>0$ and $p_{ij} = C \log n/n$ for each $i < j$, and let $d_i = \sum_{j \neq i} p_{ij}$. Suppose that
    \begin{align} \label{eq: nec-condition-m-graphs}
        \sum_{i=1}^n e^{-d_i s(1-(1-s)^{m-1})}= \Omega(1),
    \end{align}
    or equivalently that $Cs(1-(1-s)^{m-1}) < 1$. Then exact recovery is impossible, i.e. the output $\pihat_{12},\cdots,\pihat_{1m}$ of any estimator satisfies 
    \[\P{\pihat_{12} = \pistar_{12},\cdots,\pihat_{1m} = \pistar_{1m}} = 1 - \Omega(1).\] 
\end{theorem}

In the setting of ER graphs with $p_{ij} = C \log(n)/n$, Theorems~\ref{thm: achievability-m-graphs} and~\ref{thm: impossibility-m-graphs} establish that exact recovery is possible if and only if $Cs(1-(1-s)^{m-1}) > 1$. When $m=2$, this condition reduces to $Cs^2>1$, recovering the reconstruction threshold of Cullina and Kiyavash~\citep{cullina2017exact}. Thus, when $m>2$, there is a regime given by
\[Cs(1-(1-s)^{m-1})>1>Cs^2,\]
where exact recovery is impossible using only the two graphs but possible using all $m$ graphs. This is illustrated in Figure~\ref{fig: rec-regions}.


\begin{figure}[t]
    \centering
    \subfigure[$m=3$]{
    \includegraphics[width = 0.48\textwidth]{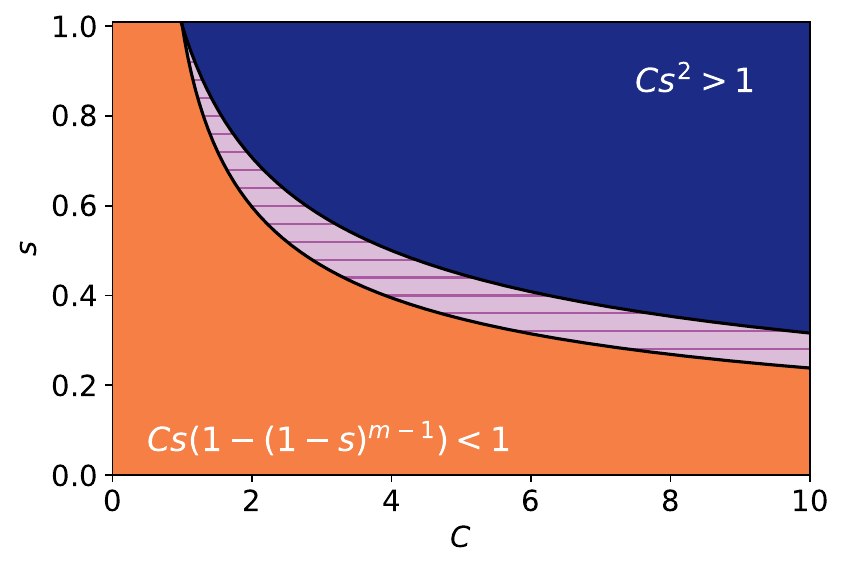}
    } \hfill
    \subfigure[$m=10$]{
    \includegraphics[width = 0.48\textwidth]{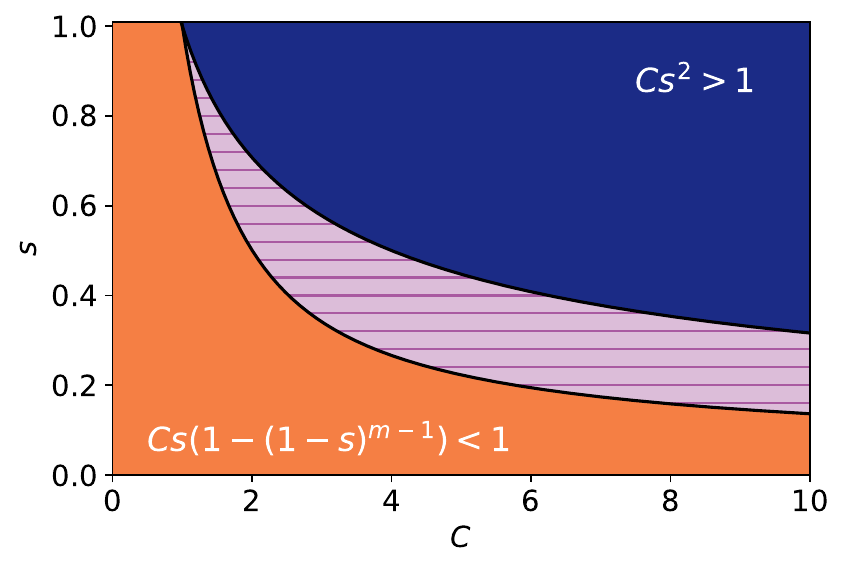}
    }
    \caption{Regions in parameter space for the homogeneous setting. \textit{Orange}: Exactly matching $m$ graphs is impossible even with $m$ graphs. \textit{Blue}: Exactly matching $2$ graphs is possible with $2$ graphs.  \textit{Striped}: Impossible to match $2$ graphs using only the $2$ graphs, but possible using $m$ graphs as side information.}
    \label{fig: rec-regions}
\end{figure}

We turn next to negative results in the heterogeneous case. To that end, consider the following lemma about isolated nodes in an inhomogeneous random graph.

\begin{lemma}
\label{lem: impossibility-two-graphs}
    Let $H\sim \IHG(n,\bfQ)$, and let $\newd_i = \sum_{j\neq i} \newp_{ij}$. Suppose that $\newp_{\max} = o(1/\sqrt{n})$ and 
    \begin{align}
    \pbr{\sum_{i=1}^{n} e^{-\newd_i}} - \max_{j\in [n]}e^{-\newd_j} = \Omega(1). \label{eq: nec-condition-2-graphs-general}
    \end{align}
    With probability bounded away from $0$, there are at least two isolated nodes in $H$.
\end{lemma}

\Cref{lem: impossibility-two-graphs} is proved in Appendix~\ref{apxsub: impossibility-two-graphs}. Applying it to $H= G_1$ yields a regime where there are at least two isolated nodes $u$ and $v$ in $G_1$, with probability bounded away from $0$. Since it is impossible to distinguish between $u$ and $v$ in $G_1$, it is impossible to match them with probability greater than $1/2$. Exactly matching the $m$ graphs is therefore impossible in this regime. This implies the following converse to~\Cref{cor: informative_P}.
 
\begin{theorem} \label{thm: impossibility-isolated-G1}
    Consider the model of $m$ correlated graphs with parameters $P$ and $s$. Let $d_i  = \sum_{j\neq i} p_{ij}$. If $p_{\max} = o(1/\sqrt{n})$ and
    \begin{align} \label{eq: nec-condition-2-graphs-}
    \pbr{\sum_{i=1}^n e^{-d_i s}} - \max_{j\in [n]} e^{-d_j s} = \Omega(1),
    \end{align}
    then for any algorithm the probability of exactly matching all $m$ graphs is $1-\Omega(1)$.
\end{theorem}

\begin{remark}
    If we ask for a given $P,s,m$ whether exact recovery is possible by some algorithm, then implicitly algorithms with access to $P$ should be included because $P$ is deterministic. It remains an open problem to give the sharp recovery threshold for a given general $P,s,m$. Theorem~\ref{thm: achievability-m-graphs} provides a performance floor while Theorem~\ref{thm: impossibility-isolated-G1} provides a performance ceiling. Theorem~\ref{thm: impossibility-m-graphs} shows~\Cref{thm: achievability-m-graphs} is tight for homogeneous $P$ while~\Cref{cor: informative_P} shows~\Cref{thm: impossibility-isolated-G1} is tight for some choices of $P$.
\end{remark}
\begin{remark} \label{rem: impossibility-two-graphs}
    \Cref{lem: impossibility-two-graphs} can also be used to give a second proof of the converse result in the homogeneous setting (i.e. Theorem~\ref{thm: impossibility-m-graphs}) in the special case of $m=2$. Setting $H = G_1\wedge G_2' \sim \IHG(n,\bfP s^2)$ in~\Cref{lem: impossibility-two-graphs}, we have that if $p_{\max} = o(1/\sqrt{n})$ and
    \begin{align} \label{eq: isolated-nodes-int-graph-threshold}
    \pbr{\sum_{i=1}^n e^{-d_i s^2} } - \max_{j\in [n]} e^{-d_j s^2}= \Omega(1), 
    \end{align}
    then with probability bounded away from $0$, there are at least two nodes $u$ and $v$ that are isolated in $G_1 \wedge G_2'$, i.e. $u$ and $v$ both have disjoint sets of neighbors in $G_1$ and $G_2'$. This implies that the maximum likelihood for Erd\H{o}s-R\'{e}nyi graphs, given in Appendix~\ref{apx: MLE}, fails. In the heterogeneous setting, however, the existence of two isolated nodes in the intersection graph $G_1 \wedge G_2'$ does not imply impossibility of exact matching, because the maximum likelihood estimator no longer has this simple form.
\end{remark}

\subsection{Estimation Algorithm}
The algorithm that achieves the guarantee of Theorem~\ref{thm: achievability-m-graphs} is based on the idea of \textit{transitive closure}: First, the graphs are matched pairwise to obtain \textit{partial} matchings $\muhat_{ij}$ between each pair of graphs. These matchings can then be boosted via transitivity. The idea is formalized in Algorithm~\ref{Alg: Match}, which runs in two steps: In step 1, the $k$-core estimator, for a suitable choice of $k$, is used to pairwise match all the graphs. For any $i$ and $j$, the $k$-core estimator selects a permutation $\nuhat_{ij}$ to maximize the size of the $k$-core of $G_i \wedge G_j^{\nuhat_{ij}}$ (the $k$-core of a graph $G$ is the largest subset of vertices $\core_k(G)$ such that the induced subgraph has minimum degree at least $k$). It then outputs a matching $\muhat_{ij}$ by restricting the domain of $\nuhat_{ij}$ to $\core_k(G_i \wedge G_j^{\nuhat_{ij}})$. Note that the matchings $\muhat_{ij}$ need not be complete matchings. In step 2, these partial matchings are \textit{boosted} as follows: If a node $v$ is unmatched between two graphs $G_i$ and $G_j$, then search for a sequence of graphs $G_i, G_{k_1}, \cdots, G_{k_{\ell}}, G_j$ such that $v$ is matched between any two consecutive graphs in the sequence. If such a sequence exists, then extend $\muhat_{ij}$ to include $v$ by transitively matching it from $G_i$ to $G_j$. We emphasize the following characteristics of the algorithm.

\vskip 0.1 in 
\begin{algorithm}[t]
\DontPrintSemicolon
\SetAlgoNlRelativeSize{-1}
\caption{Matching through transitive closure \label{Alg: Match}}
    \SetKwInOut{Inputs}{require}
    \BlankLine
    \Inputs{ Graphs $G_1, G_2,\cdots,G_m$ on a common vertex set $[n]$, Integer $k$}
    \BlankLine
    \mycomment{\textcolor{blue}{Step 1: Pairwise matching using the $k$-core estimator}}
    \BlankLine
    \For{$\{i,j\}$ in $\{1,\cdots,m\}$ such that $i < j$}{
        \BlankLine
        $\nuhat_{ij} \leftarrow \argmax_{\pi} |\core_k\pbr{G_i \wedge G_j^{\pi}}|$\;
        $\muhat_{ij} \leftarrow \nuhat_{ij} \text{ with domain restricted to }\core_k(G_i \wedge G_j^{\nuhat_{ij}})$\;
    }
    \mycomment{\textcolor{blue}{Step 2: Boosting through transitive closure}}
    \BlankLine
    \For{$v \in [n]$}{
        \For{$j = 2,\cdots,m $}{
        \If{ \text{there is a sequence of indices $1=k_1,\cdots,k_{\ell}=j$} in $[m]$ such that $\muhat_{k_{\ell-1},j}\circ \cdots \circ \muhat_{k_2,k_3} \circ \muhat_{1,k_2}(v) = v'$ for some $v' \in [n]$}{ 
        \BlankLine
        Set $\pihat_{1j}(v) = v'$\;
        }
        }
    }
    \BlankLine
    \textbf{return} $\pihat_{12},\cdots,\pihat_{1m}$
\end{algorithm}

\begin{itemize}
    \item Algorithm~\ref{Alg: Match} only ever makes \textit{pairwise} comparisons between graphs. In contrast, the maximum likelihood estimator for the simpler homogeneous setting selects the permutations $\pihat_{12},\cdots,\pihat_{1m}$ simultaneously based on all the graphs. Despite this, Algorithm~\ref{Alg: Match} works down to the threshold for exact recovery.    
    \item The transitive closure sub-routine (step 2) has low computational complexity, $O(m^2 n \log n)$. Further, it can be used to boost the pairwise matchings obtained by \textit{any} estimator in step 1, i.e. it can be used as a black-box bridge between binary graph matching and $m$-ary graph matching whenever partial pairwise matchings are available. Our choice of the $k$-core estimator for step 1 is motivated by its precision, i.e. the $k$-core estimator does not make mistakes with high probability~\citep{cullina2019kcore, racz2023matching}. This is crucial to avoid compounding errors during the boosting step. This precision is quantified formally later in the analysis (see Section~\ref{subsubsec: step 1}).
\end{itemize}



\section[Preliminaries: On the k-core of Inhomogeneous Graphs]{Preliminaries: On the $k$-core of Inhomogeneous Graphs} \label{sec-prelim-m-graphs}
    This section presents useful preliminaries and lemmas for the proof of Theorem~\ref{thm: achievability-m-graphs}. The results are stated for general inhomogeneous graphs and may be of interest more generally beyond the scope of the graph alignment problem. The main objective of this section is to establish that under appropriate conditions, all nodes with sufficiently large degree in an inhomogeneous graph are in its $k$-core with high probability. 

\begin{lemma}
\label{lemma: deg-to-kcore}
    Let $G \sim \IHG(n,\mathbf{P})$ and $d_i = \sum_{j\neq i}p_{ij}$. Suppose there exist $\alpha \in (0,1]$ and $\eps>0$ such that $p_{\max} = o(n^{\alpha/2-\eps - 1})$ and
    \[ \sum_{i=1}^{n} e^{-d_i} = o(n^{-\alpha}).\]
    Then, for any positive integer $k$ and any $v \in [n]$,
    \alns{ 
    \P{\cbr{v\notin \core_k(G)} \cap \cbr{\delta_G(v) \geq k+2/\alpha}} = o(1/n).
    }
\end{lemma}

We build up to the proof of this lemma by establishing some results about the number of low-degree nodes (in Section~\ref{subsec: low-deg}), and about the size of the $k$-core in inhomogeneous random graphs (in Section~\ref{subsec: k-core-size}). These results are then used to prove Lemma~\ref{lemma: deg-to-kcore} in Section~\ref{subsec: Proof-deg-to-kcore}.

\subsection{On low degree nodes in an inhomogeneous graph} \label{subsec: low-deg}

\begin{lemma} \label{lem: stoch-dom-of-Poissons}
    Let $p \in [0,1]$. Consider independent random variables $X$, $Y$ and $Z$ with distributions 
    \alns{ 
    X\sim \mathrm{Bern}(p), \qquad Y\sim \mathrm{Pois}(p), \qquad Z\sim 2\cdot \mathrm{Pois}(p/2).
    }
    Then, for all $t$, the moment generating functions of $X$, $Y$ and $Z$ are ordered as \[M_X(t) \leq M_Y(t) \leq M_Z(t).\]
    
\end{lemma}
\textit{Proof. }
    We have
    \alns{ 
    M_X(t) = 1+p(e^{t}-1) \stackrel{\text{(a)}}{\leq} \exp\pbr{p(e^t-1)} = M_Y(t) \stackrel{\text{(b)}}{\leq} \exp\pbr{p\pbr{\frac{e^{2t}-1}{2}}} = M_Z(t),
    }
    where (a) uses $1+x\leq e^x$ and (b) uses that the function $\pbr{e^{kt}-1}/k$ is increasing in $k$ for all $t$.\hfill $\Box{}$
    
\begin{lemma}\label{lem: stoch-dom-of-Poissons-v2}
    Let $n>0$ and let $X_1,\cdots,X_n$ be independent random variables with $X_i \sim \text{Bern}(p_i)$. Let $m$ be an integer so that $0\leq m\leq n$, and define 
    \alns{ 
    X:= \sum_{i\leq m} X_i + \sum_{j>m}2X_j.
    }
    Then, for any $t< \mathbb{E}X$,
    \alns{ 
    \P{X\leq t} \leq \exp\pbr{-\frac{\mathbb{E}X}{2} + \frac{t}{2} \log\pbr{\frac{e\cdot \mathbb{E}X}{t}}}
    }
\end{lemma}
\textit{Proof. }
    Let $Y_i \sim \text{Pois}(p_i)$ and $Z_i \sim 2\cdot \text{Pois}(p_i/2)$. It follows that
    \alns{ 
    M_X(t) &= \prod_{i=1}^{m} M_{X_i}(t) \cdot \!\!\prod_{j=m+1}^{n} \!\! M_{X_j}(2t) \\
    & \leq \prod_{i=1}^{m}M_{Z_i}(t)\cdot \!\!\prod_{j=m+1}^{n}\!\! M_{Y_j}(2t) \\
    &= \prod_{i=1}^{m} M_{Z_i}(t) \cdot \!\!\prod_{j=m+1}^{n} \!\! M_{2Y_j}(t) 
     = M_Z(t),
    }
    where $Z := \sum_{i=1}^{m}Z_i + \sum_{j=m+1}^{n}2Y_j$ is distributed as $Z\sim 2\cdot \text{Pois}\pbr{\mathbb{E}X/2}$.

    \vskip 0.1in
    \noindent Therefore, Chernoff tail bounds that hold for $Z$ also hold for $X$. In particular, for any $t<\mathbb{E}X$,
    \alns{ 
    \P{X \leq t} \leq \inf_{\theta \geq 0} e^{\theta t}M_X(-\theta) 
    \leq \inf_{\theta \geq 0} e^{\theta t} M_Z(-\theta) = \inf_{\theta\geq 0} \exp\pbr{\theta t + \frac{\mathbb{E}{X}}{2}\pbr{e^{-2\theta}-1} }.
    }
    Setting $\theta = \frac{1}{2}\log\pbr{\mathbb{E}X/t}$ gives the desired result. \hfill $\Box{}$

\begin{lemma} \label{lem: low-deg-nodes-are-few}
    Let $G\sim \IHG(n,\mathbf{P})$, and let $d_i = \sum_{j\neq i}p_{ij}$. Let $r>0$ be an integer and let $Z_r$ denote the set of nodes of $G$ with degree less than or equal to $r$. Suppose that for some $\alpha > 0$, 
    \aln{ 
    \sum_{i=1}^{n} e^{-d_i}= o(n^{-\alpha}). \label{eq: nec-degs-less-r}
    }
    \begin{itemize}
        \item[--] For any $\delta > 1-\alpha/2$,
    \aln{
    \P{|Z_r| > n^{\delta}} = o(1/n). \label{eq: Prob-Zr}
    }
    \item[--] If $p_{\max} = o(n^{\alpha/2-1})$, then~\eqref{eq: Prob-Zr} is true for all $\delta>1-\alpha$.
    \end{itemize}
    
\end{lemma}
\textit{Proof. }
If strictly more than $n^{\delta}$ nodes all have degree less than or equal to $r$, then a subset of exactly $n^{\delta}$ nodes has average degree less than $r$, i.e.
\aln{ 
\P{|Z_r|> n^{\delta}} &= \P{\exists \ T \subseteq [n] \colon \cbr{|T| \geq n^{\delta}} \cap \cbr{ \max_{u\in T} \delta_G(u) \leq r}} \nonumber \\
&\leq \P{\exists \ T \subseteq [n] \colon \cbr{|T| = n^{\delta}} \cap \cbr{ \max_{u\in T} \delta_G(u) \leq r}} \nonumber \\
& \leq \P{\exists \ S \subseteq [n]  \colon \cbr{|S| = n^{\delta}} \cap \cbr{\sum_{u\in S} \delta_{G}(u) \leq r|S|} }.\label{eq: Zr-upper-bound}
}
Let $G_{uv}$ denote the indicator random variable for the presence of edge $\{u,v\}$ in $G$. For any set $S\subseteq [n]$, define
\alns{ 
\mu_{S} := \sum_{i\in S} \mathbb{E}[\delta_G(i)] = \sum_{i\in S} d_i.
}
Since for any $i \in [n]$, we have that
\(
e^{-d_i} \leq \sum_{j =1}^n  e^{-d_j} \leq n^{-\alpha},
\)
it follows that $d_i \geq \alpha \log n$. Therefore, for any $S \subseteq [n]$, we have that $\mu_{S} \geq \alpha |S| \log n$.
Hence,
\aln{ 
\P{\sum_{u\in S} \delta_G(u) \leq r|S|} &\!= \P{\sum_{\substack{u\in S \\ v \in [n] \setminus S}} \!\!\! G_{uv} + \!\! \sum_{\substack{u,v\in S \\ u<v}} 2 G_{uv}\leq r|S| } 
\stackrel{\text{(a)}}{\leq} \exp\pbr{ \!-\frac{\mu_{S}}{2} \!+\! \frac{r|S|}{2}\log\pbr{\frac{e \cdot \mu_{S}}{r|S|}}\! } \nonumber \\
& \stackrel{\text{(b)}}{\leq} \exp\pbr{-\frac{\alpha |S|}{2}\log(n) + \frac{r|S|}{2}\log\pbr{\frac{e \cdot \alpha \log n}{r}}} \label{eq: Zr-conc-ineq}
}
where (a) uses~\Cref{lem: stoch-dom-of-Poissons-v2} and (b) uses that the function on the right hand side of (b) is decreasing in $\mu_{S}$ on the interval $(r|S|,\infty)$, and that $r|S| \leq \alpha |S| \log n \leq \mu_{S}$ for all sufficiently large $n$. Finally, combining~\eqref{eq: Zr-upper-bound} and~\eqref{eq: Zr-conc-ineq}, and using that $|S| = n^{\delta}$ yields
\alns{ 
\P{|Z_r| > n^{\delta}} &\leq \! \! \sum_{S\subseteq [n]: |S| = n^{\delta}} \!\!\!\! \P{ \sum_{u\in S} \delta_G(v) \leq r|S|}  \leq \binom{n}{n^{\delta}} \exp\pbr{- \frac{\alpha n^{\delta}}{2} \log n + \frac{r n^{\delta}}{2}\log\pbr{\frac{e\cdot \alpha \log n}{r}} \!} \\
& \stackrel{\text{(a)}}{\leq} \exp\pbr{ \pbr{1-\delta - \frac{\alpha}{2}}n^{\delta}\log n  + \frac{r n^{\delta}}{2}\log\pbr{\frac{e\cdot\alpha\log n}{r}}} =o(1/n),
}
whenever $\delta>1-\alpha/2$. Here, (a) uses the fact that $\binom{n}{n^{\delta}} \leq \exp\pbr{(1-\delta)n^{\delta}\log n}$. Hence,~\eqref{eq: Prob-Zr} follows. 

\vskip 0.1 in
\noindent Suppose that $p_{\max} = o(n^{\alpha/2-1})$. Notice that
\[ 
\sum_{u\in S} \deg_G(u) = \sum_{u\in S} \sum_{\substack{v \in [n] \\ v \neq u}} G_{uv} \leq \sum_{u\in S} \sum_{\substack{v \in [n] \setminus S}} G_{uv} \triangleq Y_{S},
\]
and so we have that 
\aln{ 
\P{\sum_{u\in S} \deg_G(u) \leq r |S|} \leq \P{Y_{S} \leq r|S|} \stackrel{\text{(a)}}{\leq} e^{-\E{Y_{S}}} \cdot \pbr{\frac{e\cdot \E{Y_{S}}}{r|S|}}^{r|S|}, \label{eq: mono}
}
where (a) uses that $Y_{S}$ is a sum of independent Bernoulli random variables, and so the Chernoff tail bound~\eqref{eq: Mitzenmacher-UB1} can be applied. On the other hand,
\[ 
\E{Y_{S}} = \sum_{u \in S} d_u - \sum_{u \in S}\sum_{\substack{v \in S \\ v \neq u}} 2 p_{uv} \stackrel{\text{(a)}}{\geq} \sum_{u\in S} d_u - \binom{|S|}{2} \cdot (2 p_{\max}) \stackrel{\text{(b)}}{\geq} \alpha |S| \log(n) - |S|^2 \cdot p_{\max} \triangleq \tmu_{S},
\]
where (a) uses that $p_{uv} \leq p_{\max}$ and (b) uses that $\mu_{S} \geq \alpha |S| \log n$. Continuing from~\eqref{eq: mono}, we have
\aln{ 
\P{Y_{S} \leq r|S|} \leq e^{-\E{Y_{S}}} \cdot \pbr{\frac{e\cdot \E{Y_{S}}}{r|S|}}^{r|S|} \leq e^{-\tmu_{S}} \cdot \pbr{\frac{e \cdot \tmu_{S}}{r|S|}}^{r|S|}, \label{eq: PYs2}
}
where the last step uses that the function $x \mapsto e^{-x} \cdot \pbr{\frac{e \cdot x}{r|S|}}^{r|S|}$ is decreasing in $x$ whenever $x \geq r|S|$. Thus, for any $S$ such that $|S| = n^{\delta}$, we have from~\eqref{eq: PYs2} that
\begin{align}
    \P{Y_{S} \leq rn^{\delta}} &\leq \exp\pbr{-n^{\delta} \pbr{\alpha \log(n) - r - n^{\delta} p_{\max}}} \cdot \pbr{ \frac{\alpha \log(n) - n^{\delta} p_{\max}}{r}}^{rn^{\delta}} \nonumber \\
    & = \exp\pbr{-n^{\delta} \pbr{\alpha \log(n) - r -n^{\delta}p_{\max} - r\log\pbr{\frac{\alpha \log(n) - n^{\delta}p_{\max}}{r}}}}.
\end{align}
Combining this with~\eqref{eq: Zr-upper-bound}, a union bound yields that
\begin{align}
    \P{|Z_r| > n^{\delta}} & \!\leq \P{\exists \ S \subseteq [n]: \cbr{|S| = n^{\delta}} \cap \cbr{\sum_{u \in S} \deg_G(u) \leq r|S|}} \leq \binom{n}{n^{\delta}} \P{Y_{S} \leq r n^{\delta}} \nonumber \\
    & \stackrel{\text{(a)}}{\leq} \!\exp\!\pbr{-n^{\delta} \Big[(\alpha\!+\!\delta\!-\!1) \log(n) - n^{\delta}p_{\max} - r \log\Big(\frac{\alpha\log(n) \!-\! n^{\delta}p_{\max}}{r} \Big)-r-1 \Big] }, \label{eq: rhs-fin}
\end{align}
where (a) uses that $\binom{n}{n^{\delta}} \leq \exp\pbr{(1-\delta)n^{\delta}\log n}$. Finally, note that the right-hand side of~\eqref{eq: rhs-fin} is $o(1/n)$ whenever $\alpha + \delta > 1$ and $n^{\delta} p_{\max} = o(\log n)$. Both conditions are satisfied when $p_{\max} = o(n^{\alpha/2-1})$ by choosing $\delta$ such that $1-\alpha < \delta < 1-\alpha/2$. \hfill $\Box{}$

\subsection[On the size of the k-core of an inhomogeneous random graph]{On the size of the $k$-core of an inhomogeneous random graph} \label{subsec: k-core-size}

The following concentration inequalities for a sum of Bernoulli random variables are used frequently in the analysis. 

\begin{lemma} \label{lem: concentration-of-binomial}
Let $X$ be a sum of $n$ independent Bernoulli random variables, and let $\mathbb{E}[X] = np$. Then, 
\begin{enumerate}
    \item For any $\delta > 0$, 
    \aln{ 
    \P{X \geq (1+\delta)np} \leq \pbr{\frac{e^\delta}{(1+\delta)^{1+\delta}}}^{np} \leq \pbr{\frac{e}{1+\delta}}^{(1+\delta)np}. \label{eq: Mitzenmacher-LB}
    }
    \item For any $\delta > 5$,
    \aln{ 
    \P{X \geq (1+\delta)np} \leq 2^{-(1+\delta)np}. \label{eq: Mitzenmacher-2B}
    }
    \item For any $\delta \in (0,1)$,
    \aln{ 
    \P{X \leq (1-\delta)np} \leq \pbr{\frac{e^{-\delta}}{(1-\delta)^{1-\delta}}}^{np} = e^{- np} \cdot \pbr{\frac{e}{1-\delta}}^{(1-\delta)np}. \label{eq: Mitzenmacher-UB1}
    }
\end{enumerate}
\end{lemma}
\textit{Proof. }
All proofs follow from the Chernoff bound and can be found, or easily derived, from Theorems 4.4 and 4.5 of~\citep{mitzenmacher2017}.
\hfill $\Box{}$

Let $G\sim\IHG(n,\mathbf{P})$ and let $k>0$. Let $v$ be any node of $G$ and consider the induced subgraph $G-v$ of $G$ on the vertex set $[n]\setminus \{v\}$. The objective of this section is to show that the $k$-core of $G-v$ is sufficiently large with probability $1-o(1/n)$. Formally, we wish to prove the following.

\begin{lemma} \label{lem: size-of-k-core-of-G-v}
    Let $G\sim \IHG(n,\mathbf{P})$, and let $v$ be a node of $G$. Suppose there exist $\alpha \in (0,1]$ and $\eps >0$ such that $
    p_{\max}= o(n^{\alpha/2 -\eps-1})
    $ and 
    \[ \sum_{i=1}^{n}e^{-d_i}= o(n^{-\alpha}).
    \]
    Then for any $\delta > 1-\alpha$, the size of the $k$-core of $G-v$ satisfies
    \alns{ 
        \P{ \left| \core_k(G-v)\right| < n - 3n^{\delta}} = o(1/n).
    }
\end{lemma}

\vskip 0.1 in 
\begin{algorithm}[t] \label{alg: Luczak}
\DontPrintSemicolon
\SetAlgoNlRelativeSize{-1}
\caption{Łuczak expansion}
    \SetKwInOut{Inputs}{require}
    \BlankLine
    \Inputs{ Graph $G$, Set $U \subseteq V(G)$.}
    \BlankLine
    $U_0 \leftarrow U$\;
    \For{$i = 0,1,2,3, \cdots$}{
        \uIf{ \text{ there exists $u \in V\setminus U_i$ such that $u$ has $3$ or more neighbors in $U_i$}}{
            $U_{i+1} \leftarrow U_i \cup \{u\}$\;
        }
        \Else{
            \textbf{return} $U_i$\;
        }
    }
\end{algorithm}

To this end, consider Algorithm~\ref{alg: Luczak} to identify a subset of the $k$-core, originally proposed by Łuczak~\citep{luczak1991size}. Note that the \textbf{for} loop eventually terminates - the set $V \setminus U_i$ is empty, for example, when $i = n$ for any input set $U$. The key is to realize that the \textbf{for} loop terminates quickly when the input $U$ is $Z_{k+1}$
, i.e. the set of vertices of the input graph $G$ whose degree is $k+1$ or less. Furthermore, for that input, the complement of the set output by the algorithm is contained in the $k$-core. This is formally stated in~\Cref{lem: Luczak} in what follows.

\begin{lemma} \label{lem: Luczak}
Let $G\sim\IHG(n,\mathbf{P})$ and let $v$ be a node of $G$. Let $U_{f}$ be the output of Algorithm~\ref{alg: Luczak} with input graph $G-v$ and  $U = Z_{k+1}$. Then,
\begin{enumerate}
    \item[\text{(a)}] $ U_f^c \subseteq \core_{k}(G-v)$.
    \item[\text{(b)}] If there exist $\alpha\in(0,1]$ and $\eps >0$ such that $p_{\max} = o(n^{\alpha/2-\eps-1})$ and $\sum_{i=1}^{n}e^{-d_i} = o(n^{-\alpha})$, then for any $\delta > 1-\alpha$, 
    \alns{ 
    \P{ |U_f| > 3n^{\delta}} = o(1/n).
    }
\end{enumerate}
\end{lemma}
\textit{Proof. }
(a) The proof is by construction: Since $U_f$ is obtained by adding exactly $f$ nodes to $U_0$, it follows that $U_f^c \subseteq U_0^c = Z_{k+1}^c$, so each node in $U_f^c$ has degree $k+2$ or more in $G-v$. Further, each node in $U_f^c$ has at most $2$ neighbors in $U_f$, else the \textbf{for} loop would not have terminated. Thus, the subgraph of $G-v$ induced on the set $U_f^c$ has minimum degree at least $k$, and the result follows.

\vskip 0.1 in
(b) Let $\delta'$ be such that $1-\alpha < \delta' < \min\cbr{\delta,1-\alpha+2\eps}$. If $|U_f| > 3n^{\delta}$, then either $|U_0| > 3n^{\delta'}$ or there is some $M$ in $\{0,1,\cdots,f\}$ for which $|U_M| = 3n^{\delta'}$. Therefore,
\alns{ 
\P{|U_f| > 3n^{\delta}} &\leq \P{|U_0| > 3 n^{\delta'}} + \P{\exists \ M \in \{0,1,\cdots,f\} \text{ s.t. } |U_M| = 3n^{\delta'}}  \\
& = o(1/n) + \underbrace{\P{\exists \ M \in \{0,1,\cdots,f\} \text{ s.t. } |U_M| = 3n^{\delta'}}}_{(\star)},
}
where the first term is $o(1/n)$ by~\Cref{lem: low-deg-nodes-are-few}. Note that each iteration $i = 0,1,\cdots,M-1$  of the \textbf{for} loop adds exactly $1$ vertex and at least $3$ edges to the subgraph of $G-v$ induced on $U_M$. Therefore, the induced subgraph $G\vert_{U_M}$ has $3n^{\delta}$ vertices and at least $3\pbr{|U_M| - |U_0|}$ edges. 

Let us couple the inhomogeneous random graph $G$ with an Erd\H{o}s-R\'{e}nyi random graph $\widetilde{G}$ on $n$ nodes with marginal distribution $\widetilde{G} \sim \text{ER}(n,p_{\max})$, such that $G$ is a subgraph of $\widetilde{G}$. Thus,
\alns{ 
(\star) & \leq \P{ \exists \text{ subgraph $H \!=\! (W,F)$ of } G-v \text{ s.t. } |W| = 3n^{\delta'} \text{ and } |F| \geq 3 \pbr{3n^{\delta'} - |U_0|} } \\
& \leq \P{|U_0| > n^{\delta'}} + \P{ \exists \text{ subgraph $H = (W,F)$ of } G \text{ s.t. } |W| = 3n^{\delta'} \text{ and } |F| \geq 6n^{\delta'} } \\
& \leq \P{|U_0| > n^{\delta'}} + \P{ \exists \text{ subgraph $\widetilde{H} \!=\! (\widetilde{W},\widetilde{F})$ of } \widetilde{G} \text{ s.t. } |\widetilde{W}| = 3n^{\delta'} \text{ and } |\widetilde{F}| \geq 6n^{\delta'} } \\
& \leq o(1/n) + \underbrace{\binom{n}{3n^{\delta'}} \cdot \P{\Bin\pbr{\binom{3n^{\delta'}}{2}, p_{\max} } > 6n^{\delta'} }}_{(\star \star)},
}
where the last step uses~\Cref{lem: low-deg-nodes-are-few} and a union bound over all possible choices of $W$. Finally, using that $\binom{n}{k}\leq \pbr{\frac{ne}{k}}^k$ and the fact that $\text{Bin}(n,p) \preceq \text{Bin}(n',p')$ whenever $n'\geq n$ and $p'\geq p$ yields
\alns{ 
(\star\star) &\leq \pbr{\frac{n^{1-\delta'} e}{3}}^{3n^{\delta'}}  \P{\Bin\pbr{\frac{9n^{2\delta'}}{2},n^{\alpha/2-\eps-1}} > 6n^{\delta'}}  \stackrel{\text{(a)}}{\leq} (n^{1-\delta'})^{3n^{\delta'}} \times \pbr{ \frac{3 e}{4} \cdot n^{\alpha/2-\eps-(1-\delta')} }^{6 n^{\delta'}} \\
& = \pbr{ \frac{3e}{4} \cdot n^{\alpha/2 -\eps - (1-\delta')/2}}^{6n^{\delta'}} = o(1/n),
}
since $\alpha/2 - \eps - (1-\delta')/2 < 0$. Note (a) uses the concentration bound~\eqref{eq: Mitzenmacher-LB} from~\Cref{lem: concentration-of-binomial}. This concludes the proof of~\Cref{lem: Luczak}.
\hfill $\Box$

Finally,~\Cref{lem: Luczak} directly implies~\Cref{lem: size-of-k-core-of-G-v}.

\subsection[Relating high degree nodes to the k-core]{Relating high degree nodes to the $k$-core: Proof of Lemma~\ref{lemma: deg-to-kcore}} \label{subsec: Proof-deg-to-kcore}

\noindent \textit{Proof of Lemma~\ref{lemma: deg-to-kcore}.}
For any set $A \subseteq [n]$, let $N_v(A)$ denote the number of neighbors of $v$ in $A$, i.e.
\alns{ 
N_v(A) := |\cbr{ u\in A: \{u,v\} \in E(G) }|.
}
Fix $v\in [n]$. Let $G-v$ denote the induced subgraph of $G$ on the vertex set $[n]-\{v\}$. Let $C_v$ denote the $k$-core of $G-v$. Since $C_v$ is a subset of the $k$-core of $G$, it follows that
\alns{ 
\cbr{v\notin \core_k(G)} \subseteq \cbr{N_v(\core_k(G)) \leq k-1} \subseteq \cbr{N_v(C_v) \leq k-1}.
}
Let $\delta \in (1-\alpha,1-\alpha+2\eps)$. It follows that
\aln{ 
\P{\cbr{v\notin \core_k(G)} \cap \cbr{\delta_G(v) \geq k+2/\alpha}} \leq q_1 + q_2, \label{eq: ref-p1-p2}
}
where 
\alns{ 
q_1:= \P{|C_v|< n -3n^{\delta}} \stackrel{\text{(a)}}{=} o(1/n),
}
where (a) follows from Lemma~\ref{lem: size-of-k-core-of-G-v}, and
\aln{ 
q_2 &:= \P{ \cbr{N_v(C_v) \leq k-1} \cap \cbr{ \delta_G(v) \geq k+ 2/\alpha} \cap \cbr{|C_v| \geq n - 3n^{1-\alpha/2}}} \nonumber \\
& \leq \sum_{\substack{A\subseteq V \\ |A| \geq n-3n^{\delta}}} \P{ \cbr{N_v(C_v) < k} \cap \cbr{ \delta_G(v) \geq k+ 2/\alpha} \mid C_v=A } \cdot \P{C_v = A} \nonumber \\
&\leq \max_{\substack{A\subseteq V \\ |A|\geq n-3n^{\delta}}} \P{\cbr{N_v(A)<k}\cap\cbr{\delta_G(v)\geq k+2/\alpha}} \nonumber
\\
& \leq \max_{\substack{A\subseteq V \\ |A|\geq n-3n^{\delta}}} \P{N_v(A^c) \geq 1+2/\alpha}  \label{eq: pause}
}
For all sets $A$ of size larger than $n-3n^{\delta}$, we have that 
\aln{ 
\P{N_v(A^c) \geq 1+2/\alpha} \!=\! \P{\sum_{i\in A^c} G_{vi} \geq 1+2/\alpha} \!\leq \P{\text{Bin}(3n^{\delta}, n^{\alpha/2-1}) \geq 1+2/\alpha}, \label{eq: pause-2}
}
where the last step uses that $\sum_{i\in A^c} G_{vi} \preceq \text{Bin}(|A^c|, p_{\max}) \preceq \text{Bin}(3n^{\delta},n^{\alpha/2-1})$. Combining~\eqref{eq: pause} and~\eqref{eq: pause-2} yields
\alns{ 
q_2 &\leq \P{\text{Bin}(3n^{\delta}, n^{\alpha/2-1}) \geq 1+2/\alpha} 
 \leq \pbr{ \frac{3 \alpha e}{2+\alpha} \cdot n^{-\alpha/2}}^{1+2/\alpha} \!= \text{const.} \times n^{-1-\alpha/2}
= o(1/n).
}
Since both $q_1$ and $q_2$ are $o(1/n)$, the result follows from~\eqref{eq: ref-p1-p2}. 
\hfill $\Box{}$.

\section[Proofs]{Proofs for Main Results} \label{sec-multiple-graphs}
    This section presents proofs for Theorems~\ref{thm: achievability-m-graphs} and~\ref{thm: impossibility-m-graphs}. Section~\ref{sec: algorithm-analysis} presents the intuition behind the analysis of Algorithm~\ref{Alg: Match}, collecting useful lemmas along the way. These lemmas are proved in Section~\ref{subsec: aux-proofs}, after the proof of Theorem~\ref{thm: achievability-m-graphs} in Section~\ref{sec: piecing-it-together}. Finally, the negative result~(\Cref{thm: impossibility-m-graphs}) is proved in Section~\ref{sec: negative-proofs}.

\subsection{Analyzing Algorithm~\ref{Alg: Match}: Intuition behind the proof of Theorem~\ref{thm: achievability-m-graphs}} \label{sec: algorithm-analysis}
Algorithm~\ref{Alg: Match} succeeds if both step 1 and step 2 succeed, i.e.
\begin{enumerate}
    \item Each instance of pairwise matching using the $k$-core estimator is \textit{precise}, i.e. it is correct on its domain with probability $1-o(1)$, i.e. \(\muhat_{ij}(v) = \pistar_{ij}(v) \) for all $v \in \dom(\muhat_{ij})$ for all $i$ and $j$.
    \item For each node $v$ and any two graphs $G_i$ and $G_j$, there is a sequence of graphs such that $v$ can be transitively matched through those graphs between $G_i$ and $G_j$.
\end{enumerate}

\subsubsection{On step 1} \label{subsubsec: step 1} This falls back to the regime of analyzing the performance of the $k$-core estimator in the setting of two graphs. Cullina and co-authors~\citep{cullina2019kcore} showed that the $k$-core estimator is \textit{precise}: For any two correlated Erd\H{o}s-R\'{e}nyi graphs $G_i$ and $G_j$ with $\bfP_{ij} =  C \log(n)/n$ for all $1\leq i< j\leq n$, and constant subsampling probability $s$, the $k$-core estimator with sufficiently large $k$ correctly matches all nodes in $\core_k(G_i' \wedge G_j')$ with probability $1-o(1)$. Building on this analysis, Racz and Sridhar~\citep{racz2023matching} showed that this is true more generally:

\begin{lemma}[Lemma III.4 from~\citep{racz2023matching}] \label{lem: racz}
Let $G_1$ and $G_2$ be correlated inhomogeneous graphs obtained by subsampling each edge independently with probability $s$ from a parent graph $G \sim \IHG(n,\bfP)$. Suppose that there exists some $\alpha \in (0,1]$ and $\eps > 0$ such that $p_{\max} \leq o(n^{\alpha/2 - \eps - 1})$, and let $k > \frac{12}{1-\alpha + 2\eps}$. Then, with probability $1-o(1)$, it holds that the matching $\muhat_k$ output by the $k$-core estimator satisfies
\[ 
\dom(\muhat_k) = \core_k(G_1' \wedge G_2') \text{ and } \muhat_k(v) = \pistar(v) \ \forall \ v \in \dom(\muhat_k).
\]
\end{lemma}

This precision of the $k$-core estimator is critical to the success of Algorithm~\ref{Alg: Match} since it ensures that there are no conflicts during the boosting step. Applying Lemma~\ref{lem: racz} to the setting of $m$ graphs yields that each instance of pairwise matching in step $1$ of Algorithm~\ref{Alg: Match} is indeed precise.

\begin{corollary} \label{cor: pairwise-is-correct}
    Let $\muhat_{ij}$ denote the matching output by the $k$-core estimator on graphs $G_i$ and $G_j$. Under the conditions of~\Cref{lem: racz},
    \alns{ 
    \mathbb{P}(\exists\ 1 \leq i < j \leq m, \text{ and } v \in \core_k(G_i'\wedge G_j') \text{ such that } \muhat_{ij}(v) \neq \pistar_{ij}(v) ) = o(1).
    }
\end{corollary}
\textit{Proof.} Since the number of instances of pairwise matchings is constant whenever $m$ is constant, a union bound yields 
\alns{ 
\mathbb{P}(\exists\ 1 &\leq i < j \leq m, \text{ and } v \in \core_k(G_i'\wedge G_j') \text{ such that } \muhat_{ij}(v) \neq \pistar_{ij}(v) )\\
& \leq \sum_{i=1}^{m} \sum_{j=1}^{m} \P{\muhat_{ij}(v) \neq \pistar_{ij}(v) \text{ for some } v \in \core_k(G_i'\wedge G_j')}  = o(1). 
}
This concludes the proof. \hfill $\Box{}$

\subsubsection{On step 2} The challenging part of the proof is to show that boosting through transitive closure matches all the nodes with probability $1-o(1)$ under the sufficient condition. It is instructive to visualize this using \textit{transitivity graphs}. 

\begin{definition}[Transitivity graph, $\calH(v)$]
For each node $v\in [n]$, let $\calH(v)$ denote the graph on the vertex set $\{g_1,\cdots,g_m\}$ such that an edge $\{g_i,g_j\}$ is present in $\calH(v)$ if and only if $v \in \core_k(G_i'\wedge G_j')$.     
\end{definition}

On the event that each instance of pairwise matching using the $k$-core is correct, the edge $\{g_i, g_j\}$ is present in $\calH(v)$ if and only if $v$ is correctly matched using the $k$-core estimator between $G_i$ and $G_j$, i.e. $\pistar_{1i}(v)$ is matched to $\pistar_{1j}(v)$. Thus, in order for step 2 to succeed (i.e. to exactly match all vertices across all graphs), it suffices that the graph $\calH(v)$ be connected for each node $v\in [n]$. However, studying the connectivity of the transitivity graphs is challenging because in any graph $\calH(v)$, no two edges are independent. This is because the $k$-cores of any two intersection graphs $G_{a}'\wedge G_{b}'$ and $G_{c}' \wedge G_{d}'$ are correlated, because all the graphs $G_a, G_b, G_c$ and $G_d$ are themselves correlated. 
To overcome this, we introduce another graph $\Ht(v)$ that relates to $\calH(v)$ and is amenable to analysis. Let $\deg_G(v)$ denote the degree of a node $v$ in graph $G$.

\begin{definition}
For each node $v\in [n]$, let $\Ht(v)$ denote a complete \textit{weighted} graph on the vertex set $\{g_1,\cdots,g_m\}$ such that the weight on any edge $\{g_i,g_j\}$ is $ \ct_v\pbr{i,j} : =\deg_{G_i' \wedge G_j'}(v).
$
\end{definition}

The relationship between the graphs $\calH(v)$ and $\Ht(v)$ stems from~\Cref{lemma: deg-to-kcore} which relates the degree of node $v$ in $G_i'\wedge G_j'$ to the inclusion of $v$ in $\core_k(G_i'\wedge G_j')$ for each $i$ and $j$. Since the graph $G_i'\wedge G_j' \sim \IHG(n, \bfP s^2)$,~\Cref{lemma: deg-to-kcore} implies that with probability $1-o(1/n)$, if a pair $\{g_i,g_j\}$ has edge weight $\ct_{ij} \geq k + 2/\alpha$ in $\Ht(v)$, then the corresponding edge $\{g_i,g_j\}$ is present in the transitivity graph $\calH(v)$, i.e. $v$ is correctly matched between $G_i$ and $G_j$ in the instance of pairwise $k$-core matching between them. The graph $\calH(v)$ is not connected only if it contains a (non-empty) vertex cut $U \subset \{1,\cdots, m\}$ with no edge crossing between $U$ and $U^c$. For a cut $U$, let $c_v(U)$ denote the number of such crossing edges in $\calH(v)$. Furthermore, define the \textit{cost} of the cut $U$ in $\Ht(v)$ as
\alns{ 
\ct_v(U):=\sum_{i\in U}\sum_{j \in U^c} \ct_v\pbr{i,j}.
}

Lemma~\ref{lemma: deg-to-kcore} is a statement about a single graph, but it can be invoked to prove the following. 

\begin{restatable}{lemma}{thmproxy} \label{lem: proxy-main-f}
    Let $G_1,\cdots,G_m$ be correlated inhomogeneous ranodm graphs from the subsampling model with parameters $\bfP$ and $s$. Suppose there exist $\alpha \in (0,1]$ and $\eps > 0$ such that $p_{\max} =o(n^{\alpha/2- \eps -1})$. Let $v \in [n]$ and let $U$ be a vertex cut of $\{1,\cdots,m\}$ such that $|U| \leq \lfloor m/2\rfloor$. Then, 
    \aln{ 
    \P{ \cbr{c_{v}(U) = 0} \cap \cbr{\ct_{v}(U) > \frac{m^2}{4}\pbr{k+\frac{2}{\alpha }}} } = o(1/n). \label{eq: int-event}
    }
\end{restatable}
The proof of~\Cref{lem: proxy-main-f} is deferred to Section~\ref{subsec: aux-proofs}. It suffices to analyze the probability that the graph $\Ht(v)$ has a cut $U$ such that its cost $\ct_v(U)$ is too small. To that end, we show that the bottleneck arises from vertex cuts of small size. Formally,

\begin{restatable}{lemma}{thmstochdom} \label{lem: summed-stoch-dom}
    Let $G_1,\cdots,G_m$ be correlated inhomogeneous graphs obtained from the subsampling model. Let $v \in [n]$ and let $U_{\ell}$ denote the set $\{1,\cdots,\ell\}$ for $\ell$ in $\{1,\cdots,\lfloor m/2\rfloor\}$. For any vertex cut $U$ of $\{1,\cdots,m\}$, let $\ct_v(U)$ denote its cost in the graph $\Ht(v)$. The following stochastic ordering holds:
    \alns{ 
    \ct_v(U_{1}) \preceq \ct_v(U_{2}) \preceq \cdots \preceq \ct_v(U_{\lfloor m/2\rfloor}).
    }
\end{restatable}

The proof of Lemma~\ref{lem: summed-stoch-dom} is deferred to Appendix~\ref{apx: stoch-dom} due to space constraints. Lemmas~\ref{lem: proxy-main-f} and~\ref{lem: summed-stoch-dom} imply that the tightest bottleneck to the connectivity of $\calH(v)$ is the event that $\ct_v(U_1)$ is below the threshold $r := \frac{m^2}{4}\pbr{k + \frac{2}{\alpha}}$, i.e. the sum of degrees of $v$ over the intersection graphs $(G_1\wedge G_j': j = 2,\cdots,m)$ is less than $r$. This event occurs only if the degree of $v$ is simultaneously less than $r$ in each of the intersection graphs $(G_1 \wedge G_j': j = 2,\cdots,m)$. The last ingredient to the analysis is that under the condition $\sum_{i=1}^{n}e^{-d_i s\pbr{1-(1-s)^{m-1}}} = o(n^{-\alpha})$, this event occurs with probability $o(1/n)$.

\begin{restatable}{lemma}{lempunchline}\label{lem: punchline}
Let $G_1,\cdots,G_m$ be obtained from the subsampling model with parameters $\bfP$ and $s$. Let $d_i =\sum_{j\neq i}p_{ij}$ and $v^* \in \argmin_{v\in V} d_v$. Let $k \geq \frac{12}{1-\alpha+2\eps}$ be constant and let $r \triangleq \frac{m^2}{4}\pbr{k+\frac{2}{\alpha}}$. Suppose there exist $\alpha \in (0,1]$ and $\eps >0$ such that $p_{\max} = o(n^{\alpha/2-\eps-1})$ and 
\begin{align}
\sum_{v \in V \setminus v^*} e^{-d_v s\pbr{1-(1-s)^{m-1}}} = o(n^{-\alpha}).  \label{eq: cond-thm}
\end{align}
 Then 
\alns{ (\star) \triangleq 
\P{ \exists \ v \in V \setminus v^*: \big\{\deg_{G_1\wedge G_2'}(v) \leq r\big\}  \cap \cdots \cap \big\{\deg_{G_1\wedge G_m'}(v) \leq r\big\}} = o(1).
}
\end{restatable}

Consequently, under the sufficient condition~\eqref{eq: suff-condition-m-graphs}, the transitivity graph $\calH(v)$ is connected for each node $v$. The above ideas are used to prove Theorem~\ref{thm: achievability-m-graphs}.

\subsection{Piecing it together: Proof of Theorem~\ref{thm: achievability-m-graphs}} \label{sec: piecing-it-together}

\textit{Proof of~\Cref{thm: achievability-m-graphs}.} Let $\pihat_{12},\cdots,\pihat_{1m}$ denote the output of Algorithm~\ref{Alg: Match} with $k\geq \frac{12}{1-\alpha+2\eps}$. Let $E_1$ denote the event that Algorithm~\ref{Alg: Match} fails to match all $m$ graphs exactly, i.e.
\alns{ 
E_1 = \cbr{\pihat_{12} \ \neq \pistar_{12} } \cup \cdots \cup \cbr{\pihat_{1m} \neq \pistar_{1m}}.
}
First, we show that the output of Algorithm~\ref{Alg: Match} is correct with probability $1-o(1)$ whenever the sufficient condition holds. If the event $E_1$ occurs, then either step 1 failed, i.e. there is a $k$-core matching $\muhat_{ij}$ that is incorrect, or step 2 failed, i.e. at least one of the graphs $\calH(v)$ is not connected. Therefore,
\alns{ 
\P{E_1} &\leq \P{ \bigcup_{i,j} \   \bigcup_{v \in \core_k(G_i' \wedge G_j')} \!\!\!\! \cbr{ \muhat_{ij} \neq \pistar_{ij}} \! } \!+ \P{\bigcup_{v\in [n]} \cbr{\calH(v) \text{ is not connected}}\! } \!\leq o(1) + q,
}
where the last step uses~\Cref{cor: pairwise-is-correct}, and $q$ denotes the probability that a transitivity graph $\calH(v)$ is not connected for some $v$ in $[n]$. For each $\ell$ in the set $\{1,\cdots,\lfloor m/2\rfloor\}$, let $U_{\ell}$ denote the set $\cbr{1,\cdots,\ell}$. Then,
\alns{
q \!&=  \P{ \bigcup_{v\in [n]} \bigcup_{\ell=1}^{\lfloor m/2\rfloor} \cbr{ \exists \ U \subset \{1,\cdots,m\} : |U| = \ell \text{ and } c_v(U) = 0 }}\\
& \leq\!\sum_{\ell = 1}^{\lfloor m/2\rfloor}\! \binom{m}{\ell} \cdot \P{ \exists \ v \in [n]: c_v(U_{\ell}) = 0 } \\
& \leq\! \sum_{\ell = 1}^{\lfloor m/2\rfloor} 
 m^{\ell} \Bigg[ \P{ \exists \ v \in [n]: \ct_v(U_{\ell}) \leq \frac{m^2}{4}\big(k + 2/\alpha\big)} \\
& ~~~~~~~~~~~~~~~~~~~~~ + \P{ \exists \ v \in [n]: \cbr{c_v(U_{\ell}) = 0}\cap\big\{\ct_v(U_{\ell}) > \frac{m^2}{4}\big( k\!+\!2/\alpha\big) \big\} } \Bigg] \\
&\stackrel{\text{(a)}}{\leq}\!  \sum_{\ell=1}^{\lfloor m/2\rfloor} \! m^{\ell} \sbr{\P{ \exists \ v \in [n]: \ct_v(U_{\ell}) \leq \frac{m^2}{4}\big(k \!+\! 2/\alpha\big)\!} \!+\! o\pbr{1}} \\
& \stackrel{\text{(b)}}{\leq}\! o(1) + \sum_{\ell=1}^{\lfloor m/2\rfloor} \!m^{\ell}\cdot  \P{ \exists \ v \in [n]: \ct_v(U_{1}) \leq \frac{m^2}{4}\big(k \!+\! 2/\alpha\big)\!}  \\
& \stackrel{\text{(c)}}{\leq}\! o(1) + \sum_{\ell=1}^{\lfloor m/2\rfloor} m^{\ell} \cdot  \P{ \exists \ v\in [n]: \big\{ \deg_{G_1\wedge G_2'}(v) \leq r \big\} \cap \cdots \cap \big\{ \deg_{G_1\wedge G_m'}(v) \leq r\big\}} \stackrel{\text{(d)}}{\leq} o(1)
}
Here, (a) uses~\Cref{lem: proxy-main-f} and a union bound, and (b) uses the fact that for any $\ell \geq 2$, the random variable $\ct_v(U_{\ell})$ stochastically dominates $\ct_v(U_1)$~(\Cref{lem: summed-stoch-dom}). Lastly,~(c) uses the definition of $\ct_v(U_1)$ and~(d) uses~\Cref{lem: punchline}. Therefore, transitive closure successfully matches all nodes in $[n]$ with probability $1-o(1)$. It follows that $\P{E_2} = o(1)$. This concludes the proof.
\hfill $\Box$

\subsection{Proofs of Lemmas~\ref{lem: proxy-main-f} and~\ref{lem: punchline}} \label{subsec: aux-proofs}

\subsubsection{Proof of Lemma~\ref{lem: proxy-main-f}.} 
For any vertex cut $U$, 
\alns{ 
\cbr{\ct_v(U) > \frac{m^2}{4}\pbr{k+\frac{2}{\alpha}}} &\stackrel{\text{(a)}}{\subseteq} \cbr{\ct_v(U) > |U|\pbr{m-|U|}\pbr{k+\frac{2}{\alpha}}} \\
& = \bigg\{ \sum_{i\in U}\sum_{j\in U^c} \deg_{G_i'\wedge G_j'}(v) > |U|(m-|U|)\bigg(k+\frac{2}{\alpha}\bigg) \bigg\} \\
& \subseteq \bigcup_{i\in U} \bigcup_{j\in U^c}\cbr{ \deg_{G_i' \wedge G_j'}(v) > k+\frac{2}{\alpha}},
}
where (a) uses the fact that the maximum of a set of numbers is greater than or equal to the average. On the other hand,
\alns{ 
\cbr{c_v(U) = 0} = \bigcap_{i\in U} \bigcap_{j\in U^c} \cbr{v\notin \core_k(G_i'\wedge G_j')}.
}
Let $p_1$ denote the probability in the LHS of~\eqref{eq: int-event}. It follows from the union bound that
\alns{ 
p_1 \leq \sum_{i \in U} \sum_{j\in U^c} \P{ \cbr{v\notin \core_k(G_i'\wedge G_j')} \cap \cbr{ \deg_{G_i' \wedge G_j'}(v) > k+\frac{2}{\alpha}} } = o(1/n),
}
since for any choice of $i$ and $j$, the graph $G_i' \wedge G_j' \sim \IHG\pbr{n, \bfP s^2}$. 
\hfill $\Box{}$

\subsubsection{Proof of Lemma~\ref{lem: punchline}} \label{sec: proof-punchline}

If $s=1$, then the graphs $G_1,\cdots,G_m$ are all isomorphic, and
\alns{ 
(\star) = \P{\exists v \in V \setminus v^*: \deg_{G_1}(v) \leq r} \leq \sum_{v\in V \setminus v^*}\P{ \deg_{G_1}(v) = 0} \leq \sum_{i=1}^{n} e^{-d_i} = o(1),
}
and so the result follows. We therefore assume without loss of generality that $s < 1$. Observe that
\begin{align}
(\star) &\leq \P{\bigcup_{r_2=0}^{r} \cdots \bigcup_{r_m=0}^{r} \pbr{\exists \ v \in V\setminus v^*: \big\{\deg_{G_1 \wedge G_2'}(v) = r_2 \big\} \cap \cdots \cap \big\{\deg_{G_1\wedge G_m'}(v) = r_m \big\}}} \nonumber \\
& \leq \sum_{r_2=0}^{r} \cdots \sum_{r_m = 0}^{r} \  \ \sum_{v\in V\setminus v^*} \mathbb{E}_{D_v}\sbr{   \P{ \big\{ \deg_{G_1\wedge G_2'}(v) = r_2 \big\} \cap \cdots \cap \big\{ \deg_{G_1\wedge G_m'}(v) = r_m\big\}\mid \deg_{G_1}(v) = D_v}} \nonumber \\
& \stackrel{\text{(a)}}{=} \sum_{r_2=0}^{r} \cdots \sum_{r_m = 0}^{r} \ \ \sum_{v\in V\setminus v^*} \mathbb{E}_{D_v}\sbr{ \prod_{\ell=2}^m \binom{D_v}{r_{\ell}} s^{r_{\ell}} (1-s)^{D_v - r_{\ell}} \cdot \indicator{r_{\ell} \leq D_v} } \nonumber\\
&\stackrel{\text{(b)}}{\leq} \sum_{r_2=0}^{r} \cdots \sum_{r_m = 0}^{r} \sbr{ \prod_{\ell=2}^{m} \pbr{\frac{e}{r_{\ell}}\cdot\frac{s}{1-s}}^{r_{\ell}}} \cdot \sum_{v\in V\setminus v^*} \mathbb{E}_{D_v}\sbr{ \prod_{\ell=2}^{m} D_v^{r_{\ell}} (1-s)^{D_v}}, \label{eq: stop1}
\end{align}
where (a) uses that the degrees of node $v$ in the graphs $G_1\wedge G_2',\cdots, G_1\wedge G_m'$ are conditionally independent of the degree of $v$ in $G_1$, and that each edge of $G_1$ is independently present in $G_1\wedge G_j'$ with probability $s$. Further, (b) uses that $\binom{n}{k}\leq \pbr{ne/k}^k$ with the implicit definition $(1/0)^0 \triangleq 1$. Note that
\begin{align}
    \sum_{v\in V\setminus v^*} \mathbb{E}_{D_v} \sbr{\prod_{\ell=2}^m D_v^{r_{\ell}} (1-s)^{D_v}} &= \sum_{v\in V\setminus v^*} \pbr{ \sum_{d=0}^{n-1} \P{D_v = d}  \prod_{\ell=2}^{m} d^{r_{\ell}} (1-s)^{d}} \nonumber \\
    &= \sum_{d=0}^{n-1} \pbr{\sum_{v\in V \setminus v^*} \P{D_v = d} \cdot d^{\sum_{\ell=2}^{m}r_{\ell}} \cdot (1-s)^{(m-1)d}}. \label{eq: stop2}
\end{align}
Proceed by splitting the outer summation in~\eqref{eq: stop2} over $d$ at $d^* \triangleq n^{\frac{\alpha}{2(m-1)r}}$. We have
\begin{align}
    \sum_{d=0}^{d^*}  \sum_{v\in V \setminus v^*} &\P{D_v = d} \cdot d^{\sum_{\ell=2}^{m}r_{\ell}} \cdot (1-s)^{(m-1)d}   \leq \sum_{v\in V\setminus v^*}  (d^*)^{\sum_{\ell=2}^{m}r_{\ell}} \sum_{d=0}^{d^*} \P{D_v = d} \cdot (1-s)^{(m-1)d} \nonumber \\
    &\leq (d^*)^{\sum_{\ell=2}^{m}r_{\ell}} \sum_{v\in V\setminus v^*} \mathbb{E}_{D_v}\sbr{(1-s)^{(m-1)D_v}} \nonumber \\
    & \stackrel{\text{(c)}}{=} (d^*)^{\sum_{\ell=2}^{m}r_{\ell}} \sum_{v\in V\setminus v^*} \prod_{u \neq v}\pbr{1-p_{uv}s + p_{uv} s(1-s)^{m-1}} \nonumber \\
    & \stackrel{\text{(d)}}{\leq} (d^*)^{\sum_{\ell=2}^{m}r_{\ell}} \sum_{v\in V\setminus v^*} \exp\bigg(- \sum_{u\neq v} p_{uv}s\pbr{1-(1-s)^{m-1}} \bigg) \nonumber \\
    &\stackrel{\text{(e)}}{\leq} n^{\alpha/2} \cdot \sum_{v\in V\setminus v^*} e^{-d_v s\pbr{1-(1-s)^{m-1}}} = o(1). \label{eq: stop3}
\end{align}
Here, (c) uses the moment generating function of $D_v$, which is distributed as a sum of independent Bernoulli random variables with means $p_{uv}$ respectively. Further, (d) uses that $1-x \leq e^{-x}$ for any $x$, and finally (e) uses the definition of $d^*$ and the condition~\eqref{eq: cond-thm}. The other part of the split sum can be bounded as
\begin{align}
    \sum_{d=d^*}^{n-1}  \sum_{v\in V \setminus v^*} &\P{D_v \!=\! d} \cdot d^{\sum_{\ell=2}^{m}r_{\ell}} \cdot (1-s)^{(m-1)d} \leq \sum_{v\in V\setminus v^*} \sbr{ \max_{d\in [d^*,n]}  d^{\sum_{\ell=2}^m r_{\ell}} \cdot (1-s)^{(m-1)d}}  \P{D_v \geq d^*} \nonumber \\
    &\stackrel{\text{(f)}}{\leq} (d^*)^{\sum_{\ell=2}^{m}r_{\ell}} (1-s)^{(m-1)d^*}\cdot \sum_{v\in V \setminus v^*} \P{D_v \geq d^*} \nonumber \\
    & \stackrel{\text{(g)}}{\leq} n^{1+\alpha/2} \cdot \pbr{(1-s)^{m-1}}^{n^{\frac{\alpha}{2(m-1)r}}} = o(1), \label{eq: stop4}
\end{align}
whenever $s>0$. Here, (f) uses that the function $d \mapsto d^{\sum r_i} (a-s)^{md}$ is decreasing on the interval $[d^*,n]$ for all sufficiently large $n$. Finally, (g) uses the definition of $d^*$ and that $\P{D_v \geq d^*} \leq 1$ for all $v$. Plugging~\eqref{eq: stop3} and~\eqref{eq: stop4} in~\eqref{eq: stop2} yields that $\sum_{v\in V\setminus v^*} \mathbb{E}_{D_v} \sbr{\prod_{\ell=2}^m D_v^{r_{\ell}} (1-s)^{D_v}} = o(1)$. Therefore, continuing from~\eqref{eq: stop1}, 
\begin{align}
    \eqref{eq: stop1} &= \sum_{r_2=0}^r \cdots \sum_{r_m=0}^r \sbr{\prod_{\ell = 2}^{m} \pbr{\frac{e}{r_{\ell}} \cdot \frac{s}{1-s}}^{r_{\ell}}} \cdot o(1) \nonumber \\
    &\leq (r+1)(m-1) \cdot \max_{1\leq r' \leq r }\pbr{\frac{e s}{r'(1-s)}}^{r'(m-1)} \cdot o(1) = o(1), \nonumber
\end{align}
since $r$ and $m$ are both constants independent of $n$. This completes the proof.
\hfill $\Box$

\subsection{Proof of Theorem~\ref{thm:informative_P}}

In this section, we prove~\Cref{thm:informative_P}, which implies that our sufficient condition~\eqref{eq: stronger-suff-condition-m-graphs} is \textit{not} necessary for matching inhomogeneous graphs in general. 

\textit{Proof of Theorem~\ref{thm:informative_P}.} The proof is nonconstructive; the idea is to show that recovery is possible if the parameter matrix $P$ is random with a suitable distribution, and then show that a set of deterministic choices for $P$ satisfying the requirements of the proposition has positive probability (actually probability converging to one) under the distribution of $P.$  Let $G_{-1}$ denote an Erd\H{o}s-R\'{e}nyi graph with parameters $n$ and $q$, where $q$ is to be determined later.   Let $P$ denote the random matrix equal to $p_{\max}$ times the incidence matrix of $G_{-1}.$

 Given $P$, let $G \sim \IHG(n,P)$, let $G'_1, \cdots G'_m$ denote the graphs obtained from $G$ by subsampling each edge with probability $s$, and let $G_k$ be obtained from $G'_k$ by randomly uniformly permuting the labels, for $k\in [m]$. We will consider the algorithm that separately matches $G_{-1}$ to $G_k$ for each $k$, and then matches $G_1,\cdots, G_m$ to each other via transitive closure through $G_{-1}$. For $k\in [m]$, $G'_k$ can be viewed as being obtained from $G_{-1}$ by independently subsampling the edges of $G_{-1}$ with probability $p_{\max}s.$
 Therefore, if $r_{ij} \triangleq \P{ \indicator{e\in E(G_{-1})} = i, \indicator{e\in E(G_k)} = j}$, then
 \begin{align*}
     \left( \begin{array}{cc}
     r_{00}~~ & r_{01} \\ r_{10}~~ & r_{11}
     \end{array}
     \right)
     = \left( \begin{array}{cc}
     1-q &~~ 0 \\ q(1-p_{\max}s)  &~~ qp_{\max}s
     \end{array}
     \right)
 \end{align*}
 It is known~\citep{cullina2017exact, wu2022settling} for correlated  Erd\H{o}s-R\'{e}nyi graphs
 such as this that exact recovery is possible if there is some $\eps' > 0$ such that
\begin{align*}
    n\pbr{\sqrt{r_{00}r_{11}}-\sqrt{r_{01}r_{10}}}^2 \geq (1+\eps')\log n.
\end{align*}
Assume without loss of generality that $\epsilon < 1.$
Let $\epsilon' = \epsilon/4$ and select $q$ so that $qs=\frac{1+ 2\epsilon'}{\log n}.$ With that choice we conclude that the algorithm described exactly matches $G_k$ to $G_{-1}$ with probability $1-o(1)$ for each $k$, and hence matches $G_1,\cdots,G_m$ to each other with probability $1-o(1).$

Let $\delta_{G_{-1}}(i)$ denote the degree of node $i$ in graph $G_{-1}$. Note that $\delta_{G_{-1}}(i) \sim \Bin(n-1,q) \preceq \Bin(n,q)$ for all nodes $i$.  Therefore by the union bound and  Chernoff inequality [Theorem 4.4\citep{mitzenmacher2017}],
\begin{align}
\P{ \bigcup_{i=1}^{n} \cbr{\delta_{G_{-1}}(i) > (1+\eps')n q }  } \leq n  \cdot \P{\Bin\pbr{n,q} {>} (1+\eps') n q} \leq n e^{-nq(\eps')^2/3} = o(1). \label{eq:deg_bnd}
\end{align}
Let  $D_i = p_{\max} \cdot \delta_{G_{-1}}(i)$ for $i\in [n],$ so that $D_i$ is the random version of $d_i.$   Then \eqref{eq:ds_bound}, \eqref{eq:deg_bnd}, and the fact $(1+\epsilon')(1+2\epsilon') < 1+ \epsilon$ imply that
$ \P{ D_i \leq d ~~ \mbox{for}~ i\in [n]}=1-o(1). $
In summary, the intersection of the following three events has probability $1-o(1):$
\begin{enumerate}
    \item $\max_{i,j} P_{ij} = p_{\max} = (\log n)^2/n$.
    \item $G_1, \cdots, G_m$ are exactly matched by the transitive closure algorithm via $G_{-1}$.
    \item  $D_i \leq d ~~ \mbox{for}~ i\in [n]$.
\end{enumerate}
Therefore, sampling from the distribution of $P$ generates with positive probability, a deterministic parameter matrix $P$ satisfying 1-3 above, and for which exact graph matching is possible.
\hfill $\Box{}$

\subsection{On Necessary Conditions: Proofs of Theorem~\ref{thm: impossibility-m-graphs}} \label{sec: negative-proofs}

In this section, we prove Theorem~\ref{thm: impossibility-m-graphs}, i.e. our sufficient condition is also necessary for the homogeneous case. Theorem~\ref{thm: impossibility-m-graphs} has a simple proof following a genie-aided converse argument: reduce the problem to that of matching two graphs by providing extra information to the estimator.

\vskip 0.1 in
\noindent \textit{Proof of Theorem~\ref{thm: impossibility-m-graphs}}. If the correspondences $\pistar_{12},\cdots,\pistar_{1,m-1}$ were provided as extra information to an estimator, then the estimator must still match $G_m$ with the union graph $G_1' \vee G_2' \vee \cdots \vee G_{m-1}'$. This can be viewed as an instance of matching two inhomogeneous graphs obtained by \textit{asymmetric} subsampling: the graph $G_m$ is obtained from a parent graph $G \sim \IHG(n,\bfP)$ by subsampling each edge independently with probability $s_1 := s$, and the graph $\widetilde{G}_{m-1} \triangleq G_1'\vee G_2' \vee \cdots \vee G_{m-1}'$ is obtained from $G$ by subsampling each edge independently with probability $s_2 \triangleq 1-(1-s)^{m-1}$. 
Cullina and Kiyavash~\citep{cullina2017exact} studied this model for matching two graphs: Theorem~2 of~\citep{cullina2017exact} establishes that matching $G_m$ and $\widetilde{G}_{m-1}$ is impossible if $Cs_1s_2 < 1$, or equivalently if $Cs(1-(1-s)^{m-1}) < 1$. This is equivalent to the necessary condition~\eqref{eq: nec-condition-m-graphs}.
\hfill $\Box{}$

\section[Simulation Results]{Simulation Results and Discussion} \label{sec-simulations}
    A result of R\'acz and Sridhar~\citep{racz2023matching} establishes that with high probability, the $k$-core estimator for two inhomogeneous graphs matches only the nodes in the $k$-core of their true intersection graph. This is true even if the matrix $\bfP$ has constant entries that do not depend on $n$. Although the $k$-core estimator itself is not efficient to implement, its output can be efficiently simulated by computing the $k$-core of the true intersection graph. This is possible only because the simulator has access to the ground truth. Under the assumption that the $k$-core estimator outputs the $k$-core of the true intersection graph, we can simulate the performance of Algorithm~\ref{Alg: Match}. Figure~\ref{fig: sims} shows the mean fraction of matched nodes as a function of $m$ before and after transitive closure, for a variety of graph models. For each value of $m$, the $\binom m 2$ pairs of graphs are pairwise matched by simulating the $k$-core estimator, and the average fraction of matched nodes is recorded. The transitive closure subroutine is then applied to the partial matchings and the improvement in the fraction of matched nodes is recorded. In all figures, the shaded region represents the standard deviation for the fraction of matched nodes. Three settings are studied:

\begin{figure}[t]
    \subfigure[Erd\H{o}s-R\'{e}nyi graphs, $(n,p,s,k) = (10^4,0.003,0.8,13)$]{
    \centering
    \includegraphics[width=0.45\textwidth]{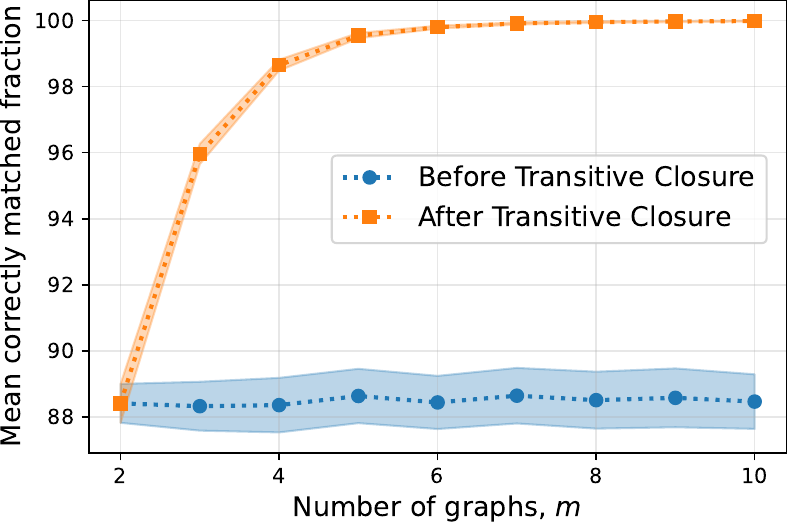}
    \label{fig: ER}
    } \hfill
    \subfigure[SBM: $(n,p,q,s,k) = (10^4, 0.04,0.01,0.25,14)$]{
    \centering
    \includegraphics[width=0.45\textwidth]{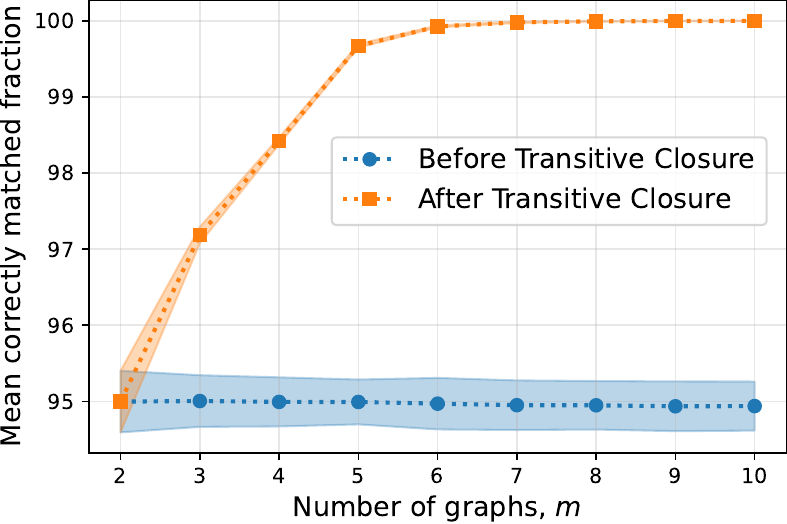}
    \label{fig: SBM}    
    }
    \vfill 
    \centering
    \subfigure[RGG, $(n,p,r,s,k) = (10^4, 0.15, 0.2,0.4,14)$]{
    \includegraphics[width = 0.45 \textwidth]{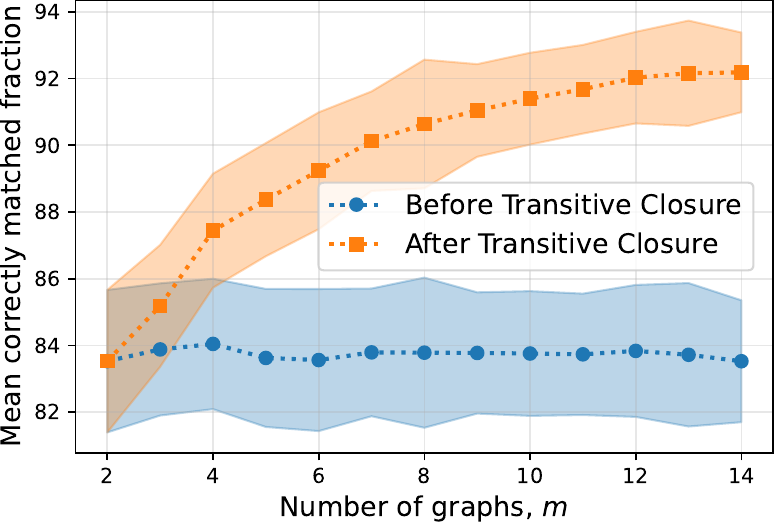} 
    \label{fig: RGG}
    }
    \caption{Mean fraction of correctly matched nodes for various graphs obtained by simulating Algorithm~\ref{Alg: Match}}
    \label{fig: sims}
\end{figure}

\begin{enumerate}
    \item Erd\H{o}s-R\'{e}nyi graphs: Figure~\ref{fig: ER} studies ER graphs on $10^4$ nodes. An edge is present in the parent graph with probability $p = 0.003$ and subsampled independently to the graphs $G_1$ and $G_2$ with probability $0.8$.
    \item Stochastic Block Model: Figure~\ref{fig: SBM} similarly presents the fraction of matched nodes before and after transitive closure for a balanced $5$-community stochastic block model on $10^4$ nodes, with intra-community edge probability $p = 0.04$ and inter-community edge probability $q=0.01$ respectively. The subsampling probability $s = 0.25$.
    \item Random Geometric Graph: Figure~\ref{fig: RGG} considers random geometric graphs generated as follows: First, $n = 10^4$ points are sampled independently and uniformly on the unit square $[0,1]\times [0,1]$. The parent graph $G$ is generated by connecting two points within Euclidean distance $r = 0.2$ with probability $p=0.15$, and the graphs $G_1$ and $G_2$ are obtained by subsampling each edge independently with probability $s = 0.4$.
\end{enumerate}

All data points are obtained by averaging across $50$ independent runs of the algorithm. As Figure~\ref{fig: sims} shows, the transitive closure step boosts the fraction of matched nodes across any two graphs.

A limitation of Algorithm~\ref{Alg: Match} is the runtime -- it does not run in polynomial time because it uses the $k$-core estimator for pairwise matching, which involves searching over the space of permutations. Even so, it is useful to establish the fundamental limits of exact recovery, and serve as a benchmark to compare the performance of any other algorithm. 

The transitive closure subroutine (Step 2) itself is \textit{efficient} because it runs in polynomial time. Therefore, a natural next step is to modify Step 1 in our algorithm so that the pairwise matchings are done by an \textit{efficient} algorithm. However, it is not clear if transitive closure is optimal for combining information from the pairwise matchings in this setting. For example, there is a possibility that the pairwise matchings resulting from the efficient algorithm are heavily correlated, and transitive closure is unable to boost them.

\section{Conclusion} \label{sec-conclusion}
    In this work, we introduced and analyzed matching through transitive closure - a black-box algorithm to combine information from multiple graphs to recover the underlying correspondence between them. Despite its simplicity, it turns out that matching through transitive closure is an optimal way to combine information in the homogeneous setting where the graphs are pairwise matched using the $k$-core estimator. The analysis of this algorithm yielded a sufficient condition to exactly match $m$ inhomogeneous random graphs. Our work presents several directions for future research.

\begin{itemize}
    \item \textbf{Impossibility results for inhomogeneous graphs.} The derivation and analysis of the maximum likelihood estimator for exact recovery can help quantify the sharpness of the sufficient condition~\eqref{eq: stronger-suff-condition-m-graphs}. Can the heterogeneity in the model be exploited to match two graphs even when their true intersection graph is disconnected? 
    \item \textbf{Polynomial-time algorithms.} Using a polynomial-time estimator in place of the $k$-core estimator in Step 1 of Algorithm~\ref{Alg: Match} yields a polynomial-time algorithm to match $m$ graphs. It is critical that this estimator is precise. Can the performance guarantees of the $k$-core estimator be realized through polynomial time algorithms that meet this constraint?
    \item \textbf{Boosting for partial recovery.} This work focused on \textit{exact} recovery, where the objective is to match \textit{all} nodes across \textit{all} graphs. It would be interesting to consider a regime where any instance of pairwise matching recovers at best a small fraction of nodes. Is it possible to quantify the extent to which transitive closure boosts the number of matched nodes?
  \item \textbf{Robustness.} Finally, how sensitive to perturbation is the transitive closure algorithm? Is it possible to quantify the extent to which an adversary may perturb edges in some of the graphs without losing the performance guarantees of the matching algorithm? Algorithms that perform well on models such as inhomogeneous graphs and are further generally robust are expected to also work well with real-world networks. 
\end{itemize}

\bibliographystyle{alpha}
\bibliography{bibliography}

\newcommand{\etalchar}[1]{$^{#1}$}
\begin{thebibliography}{KHGPM16}

\bibitem[Abb18]{abbe2018community}
Emmanuel Abbe.
\newblock Community detection and stochastic block models: recent developments.
\newblock {\em Journal of Machine Learning Research}, 18(177):1--86, 2018.

\bibitem[AH23]{ameen2023robust}
Taha Ameen and Bruce Hajek.
\newblock Robust graph matching when nodes are corrupt.
\newblock {\em arXiv preprint arXiv:2310.18543}, 2023.

\bibitem[BCL{\etalchar{+}}19]{barak2019nearly}
Boaz Barak, Chi-Ning Chou, Zhixian Lei, Tselil Schramm, and Yueqi Sheng.
\newblock ({N}early) efficient algorithms for the graph matching problem on correlated random graphs.
\newblock {\em Advances in Neural Information Processing Systems}, 32, 2019.

\bibitem[BSI06]{bandyopadhyay2006systematic}
Sourav Bandyopadhyay, Roded Sharan, and Trey Ideker.
\newblock Systematic identification of functional orthologs based on protein network comparison.
\newblock {\em Genome research}, 16(3):428--435, 2006.

\bibitem[CK16]{cullina2016improved}
Daniel Cullina and Negar Kiyavash.
\newblock Improved achievability and converse bounds for {E}rd{\H{o}}s-{R}{\'e}nyi graph matching.
\newblock {\em ACM SIGMETRICS performance evaluation review}, 44(1):63--72, 2016.

\bibitem[CK17]{cullina2017exact}
Daniel Cullina and Negar Kiyavash.
\newblock Exact alignment recovery for correlated {E}rd{\H{o}}s-{R}{\'e}nyi graphs.
\newblock {\em arXiv preprint arXiv:1711.06783}, 2017.

\bibitem[CKMP19]{cullina2019kcore}
Daniel Cullina, Negar Kiyavash, Prateek Mittal, and Vincent Poor.
\newblock Partial recovery of {E}rd{\H{o}}s-{R}{\'e}nyi graph alignment via $k$-core alignment.
\newblock {\em Proceedings of the ACM on Measurement and Analysis of Computing Systems}, 3(3):1--21, 2019.

\bibitem[CL02]{chung2002average}
Fan Chung and Linyuan Lu.
\newblock The average distances in random graphs with given expected degrees.
\newblock {\em Proceedings of the National Academy of Sciences}, 99(25):15879--15882, 2002.

\bibitem[CPPDG24]{calissano2024graph}
Anna Calissano, Theodore Papadopoulo, Xavier Pennec, and Samuel Deslauriers-Gauthier.
\newblock Graph alignment exploiting the spatial organization improves the similarity of brain networks.
\newblock {\em Human Brain Mapping}, 45(1):e26554, 2024.

\bibitem[CR24]{chai2024efficient}
Shuwen Chai and Mikl{\'o}s~Z R{\'a}cz.
\newblock Efficient graph matching for correlated stochastic block models.
\newblock {\em arXiv preprint arXiv:2412.02661}, 2024.

\bibitem[DCKG19]{dai2019canonicallabeling}
Osman~Emre Dai, Daniel Cullina, Negar Kiyavash, and Matthias Grossglauser.
\newblock Analysis of a canonical labeling algorithm for the alignment of correlated {E}rd{\H{o}}s-{R}{\'e}nyi graphs.
\newblock {\em Proceedings of the ACM on Measurement and Analysis of Computing Systems}, 3(2):1--25, 2019.

\bibitem[DD23]{ding2023densesubgraph}
Jian Ding and Hang Du.
\newblock Matching recovery threshold for correlated random graphs.
\newblock {\em The Annals of Statistics}, 51(4):1718--1743, 2023.

\bibitem[DFW23]{ding2023efficiently}
Jian Ding, Yumou Fei, and Yuanzheng Wang.
\newblock Efficiently matching random inhomogeneous graphs via degree profiles.
\newblock {\em arXiv preprint arXiv:2310.10441}, 2023.

\bibitem[DL23]{ding2023polynomial}
Jian Ding and Zhangsong Li.
\newblock A polynomial-time iterative algorithm for random graph matching with non-vanishing correlation.
\newblock {\em arXiv preprint arXiv:2306.00266}, 2023.

\bibitem[DMWX21]{ding2021degreeprofile}
Jian Ding, Zongming Ma, Yihong Wu, and Jiaming Xu.
\newblock Efficient random graph matching via degree profiles.
\newblock {\em Probability Theory and Related Fields}, 179:29--115, 2021.

\bibitem[Du25]{du2025optimal}
Hang Du.
\newblock Optimal recovery of correlated {E}rdos-{R}enyi graphs.
\newblock {\em arXiv preprint arXiv:2502.12077}, 2025.

\bibitem[FMWX22]{fan2022spectral}
Zhou Fan, Cheng Mao, Yihong Wu, and Jiaming Xu.
\newblock Spectral graph matching and regularized quadratic relaxations {II}: {E}rd{\H{o}}s-{R}{\'e}nyi graphs and universality.
\newblock {\em Foundations of Computational Mathematics}, pages 1--51, 2022.

\bibitem[GML21]{ganassali2021impossibility}
Luca Ganassali, Laurent Massouli{\'e}, and Marc Lelarge.
\newblock Impossibility of partial recovery in the graph alignment problem.
\newblock In {\em Conference on Learning Theory}, pages 2080--2102. PMLR, 2021.

\bibitem[GRS22]{gaudio2022exact}
Julia Gaudio, Mikl{\'o}s~Z R{\'a}cz, and Anirudh Sridhar.
\newblock Exact community recovery in correlated stochastic block models.
\newblock In {\em Conference on Learning Theory}, pages 2183--2241. PMLR, 2022.

\bibitem[HLL83]{holland1983stochastic}
Paul~W Holland, Kathryn~Blackmond Laskey, and Samuel Leinhardt.
\newblock Stochastic blockmodels: First steps.
\newblock {\em Social networks}, 5(2):109--137, 1983.

\bibitem[HM23]{hall2023partial}
Georgina Hall and Laurent Massouli{\'e}.
\newblock Partial recovery in the graph alignment problem.
\newblock {\em Operations Research}, 71(1):259--272, 2023.

\bibitem[HNM05]{haghighi2005robust}
Aria Haghighi, Andrew~Y Ng, and Christopher~D Manning.
\newblock Robust textual inference via graph matching.
\newblock In {\em Proceedings of Human Language Technology Conference and Conference on Empirical Methods in Natural Language Processing}, pages 387--394, 2005.

\bibitem[HSY24]{huang2024information}
Dong Huang, Xianwen Song, and Pengkun Yang.
\newblock Information-theoretic thresholds for the alignments of partially correlated graphs.
\newblock {\em arXiv preprint arXiv:2406.05428}, 2024.

\bibitem[Ind23]{GWI}
Global~Web Index.
\newblock Social behind the screens trends report.
\newblock {\em GWI}, 2023.

\bibitem[JLK21]{josephs2021recovery}
Nathaniel Josephs, Wenrui Li, and Eric.~D. Kolaczyk.
\newblock Network recovery from unlabeled noisy samples.
\newblock In {\em 2021 55th Asilomar Conference on Signals, Systems, and Computers}, pages 1268--1273, 2021.

\bibitem[KHGPM16]{kazemi2016proper}
Ehsan Kazemi, Hamed Hassani, Matthias Grossglauser, and Hassan Pezeshgi~Modarres.
\newblock Proper: global protein interaction network alignment through percolation matching.
\newblock {\em BMC bioinformatics}, 17(1):1--16, 2016.

\bibitem[LR23]{liu2023phase}
Suqi Liu and Mikl{\'o}s~Z R{\'a}cz.
\newblock Phase transition in noisy high-dimensional random geometric graphs.
\newblock {\em Electronic Journal of Statistics}, 17(2):3512--3574, 2023.

\bibitem[LS18]{lubars2018correcting}
Joseph Lubars and R~Srikant.
\newblock Correcting the output of approximate graph matching algorithms.
\newblock In {\em IEEE INFOCOM 2018-IEEE Conference on Computer Communications}, pages 1745--1753. IEEE, 2018.

\bibitem[{\L}uc91]{luczak1991size}
Tomasz {\L}uczak.
\newblock Size and connectivity of the $k$-core of a random graph.
\newblock {\em Discrete Mathematics}, 91(1):61--68, 1991.

\bibitem[MRT23]{mao2023constant}
Cheng Mao, Mark Rudelson, and Konstantin Tikhomirov.
\newblock Exact matching of random graphs with constant correlation.
\newblock {\em Probability Theory and Related Fields}, 186(1-2):327--389, 2023.

\bibitem[MU17]{mitzenmacher2017}
Michael Mitzenmacher and Eli Upfal.
\newblock {\em Probability and computing: Randomization and probabilistic techniques in algorithms and data analysis}.
\newblock Cambridge University Press, 2017.

\bibitem[MWXY23]{mao2023chandelier}
Cheng Mao, Yihong Wu, Jiaming Xu, and Sophie~H Yu.
\newblock Random graph matching at {O}tter’s threshold via counting chandeliers.
\newblock In {\em Proceedings of the 55th Annual ACM Symposium on Theory of Computing}, pages 1345--1356, 2023.

\bibitem[MX20]{mossel2020seeded}
Elchanan Mossel and Jiaming Xu.
\newblock Seeded graph matching via large neighborhood statistics.
\newblock {\em Random Structures \& Algorithms}, 57(3):570--611, 2020.

\bibitem[NS08]{narayanan2008robust}
Arvind Narayanan and Vitaly Shmatikov.
\newblock Robust de-anonymization of large sparse datasets.
\newblock In {\em 2008 IEEE Symposium on Security and Privacy (sp 2008)}, pages 111--125. IEEE, 2008.

\bibitem[NS09]{narayanan2009deanonymizing}
Arvind Narayanan and Vitaly Shmatikov.
\newblock De-anonymizing social networks.
\newblock In {\em 2009 30th IEEE Symposium on Security and Privacy}, pages 173--187. IEEE, 2009.

\bibitem[PG11]{pedarsani2011privacy}
Pedram Pedarsani and Matthias Grossglauser.
\newblock On the privacy of anonymized networks.
\newblock In {\em Proceedings of the 17th ACM SIGKDD International Conference on Knowledge Discovery and Data Mining}, pages 1235--1243, 2011.

\bibitem[RS21]{racz2021correlated}
Mikl{\'o}s~Z R{\'a}cz and Anirudh Sridhar.
\newblock Correlated stochastic block models: Exact graph matching with applications to recovering communities.
\newblock {\em Advances in Neural Information Processing Systems}, 34:22259--22273, 2021.

\bibitem[RS23]{racz2023matching}
Mikl{\'o}s~Z R{\'a}cz and Anirudh Sridhar.
\newblock Matching correlated inhomogeneous random graphs using the $k$-core estimator.
\newblock In {\em 2023 IEEE International Symposium on Information Theory (ISIT)}, pages 2499--2504. IEEE, 2023.

\bibitem[RZ24]{racz2024harnessing}
Mikl{\'o}s~Z R{\'a}cz and Jifan Zhang.
\newblock Harnessing multiple correlated networks for exact community recovery.
\newblock {\em arXiv preprint arXiv:2412.02796}, 2024.

\bibitem[SS05]{schellewald2005probabilistic}
Christian Schellewald and Christoph Schn{\"o}rr.
\newblock Probabilistic subgraph matching based on convex relaxation.
\newblock In {\em International Workshop on Energy Minimization Methods in Computer Vision and Pattern Recognition}, pages 171--186. Springer, 2005.

\bibitem[STK05]{sporns2005human}
Olaf Sporns, Giulio Tononi, and Rolf K{\"o}tter.
\newblock The human connectome: a structural description of the human brain.
\newblock {\em PLoS computational biology}, 1(4):e42, 2005.

\bibitem[SXB08]{singh2008global}
Rohit Singh, Jinbo Xu, and Bonnie Berger.
\newblock Global alignment of multiple protein interaction networks with application to functional orthology detection.
\newblock {\em Proceedings of the National Academy of Sciences}, 105(35):12763--12768, 2008.

\bibitem[VF13]{vento2013graph}
Mario Vento and Pasquale Foggia.
\newblock Graph matching techniques for computer vision.
\newblock In {\em Image Processing: Concepts, Methodologies, Tools, and Applications}, pages 381--421. IGI Global, 2013.

\bibitem[VM25]{vassaux2025feasibility}
Louis Vassaux and Laurent Massouli{\'e}.
\newblock The feasibility of multi-graph alignment: a {B}ayesian approach.
\newblock {\em arXiv preprint arXiv:2502.17142}, 2025.

\bibitem[WXY22]{wu2022settling}
Yihong Wu, Jiaming Xu, and Sophie~H Yu.
\newblock Settling the sharp reconstruction thresholds of random graph matching.
\newblock {\em IEEE Transactions on Information Theory}, 68(8):5391--5417, 2022.

\bibitem[WZWW24]{wang2024feasible}
Ziao Wang, Ning Zhang, Weina Wang, and Lele Wang.
\newblock On the feasible region of efficient algorithms for attributed graph alignment.
\newblock {\em IEEE Transactions on Information Theory}, 2024.

\bibitem[YC24]{yang2024exact}
Joonhyuk Yang and Hye~Won Chung.
\newblock Exact graph matching in correlated {G}aussian-attributed {E}rd{\H{o}}s-{R}{\'e}nyi mode.
\newblock In {\em 2024 IEEE International Symposium on Information Theory (ISIT)}, pages 3450--3455. IEEE, 2024.

\bibitem[YDF{\etalchar{+}}22]{ye2022multi}
Kai Ye, Siyan Dong, Qingnan Fan, He~Wang, Li~Yi, Fei Xia, Jue Wang, and Baoquan Chen.
\newblock Multi-robot active mapping via neural bipartite graph matching.
\newblock In {\em Proceedings of the IEEE/CVF conference on computer vision and pattern recognition}, pages 14839--14848, 2022.

\bibitem[YYL{\etalchar{+}}16]{yan2016short}
Junchi Yan, Xu-Cheng Yin, Weiyao Lin, Cheng Deng, Hongyuan Zha, and Xiaokang Yang.
\newblock A short survey of recent advances in graph matching.
\newblock In {\em Proceedings of the 2016 ACM on international conference on multimedia retrieval}, pages 167--174, 2016.

\bibitem[ZWWW24]{zhang2024attributed}
Ning Zhang, Ziao Wang, Weina Wang, and Lele Wang.
\newblock Attributed graph alignment.
\newblock {\em IEEE Transactions on Information Theory}, 2024.

\end{thebibliography}

\appendix
\section{Maximum Likelihood Estimator for Exactly Matching Multiple Homogeneous Graphs} \label{apx: MLE}

We show that in the homogeneous setting, the maximum likelihood estimator is given by the permutation profile that minimizes the number of edges in the corresponding union graph, i.e.
\[
    (\pihatMLE_{12},\cdots,\pihatMLE_{1m}) \in \argmin_{\pi_{12},\cdots,\pi_{1m}} |E(G_1 \vee G_2^{\pi_{12}} \vee \cdots \vee G_m^{\pi_{1m}})|.
\]

\begin{lemma}\label{lem: MLE-derivation}
    Let $G_1,\cdots,G_m$ be correlated Erd\H{o}s-R\'{e}nyi graphs obtained from the subsampling model with parameters $p$ and $s$. Further, let $\pi_{12},\cdots,\pi_{1m}$ denote a collection of permutations on $[n]$. Then
    \alns{ 
    \log \P{G_1,\cdots,G_m \mid \pistar_{12} = \pi_{12},\cdots,\pistar_{1m} = \pi_{1m}} \propto \mathrm{const}. - |E(G_1 \vee G_2^{\pi_{12}} \vee \cdots \vee G_m^{\pi_{1m}})|,
    }
    where $\mathrm{const}.$ depends only on $p$, $s$ and $G_1,\cdots,G_m$.
\end{lemma}

\noindent \textit{Proof.} Consider the setting where $p_{ij} = p$ for all $i\neq j$. Notice that
\aln{
\P{G_1,\cdots,G_m \vert \pistar_{12},\cdots,\pistar_{1m}} &= \prod_{e \in \binom{[n]}{2}} \P{G_1(e),G_2(\pistar_{12}(e))\cdots,G_m(\pistar_{1m}(e))\mid \pistar_{12},\cdots,\pistar_{1m}}\nonumber \\
& = \prod_{e \in \binom{[n]}{2}} \P{G_1(e),G_2'(e)\cdots,G_{m}'(e) \label{eq: edge-prob}}
}
where for a node pair $e = \{u,v\}$, the shorthand $\pi(e)$ denotes $\{ \pi(u), \pi(v)\}$. The edge status of any node pair $e$ in the graph tuple $(G_1,G_2',\cdots,G_m')$ can be any of the $2^m$ bit strings of length $m$, but the corresponding probability in~\eqref{eq: edge-prob} depends only on the number of ones and zeros in the bit string. For $i \in [m]$, let $\alpha_i$ denote the number of node pairs $e$ whose corresponding tuple $(G_1(e),G_2'(e),\cdots,G_m'(e)) $ has exactly $i$ 1's:
\alns{ 
\alpha_i := \sum_{e\in \binom{[n]}{2}} \1\cbr{(G_1(e),G_2'(e),\cdots,G_m'(e)) \text{ has exactly } i \ 1\text{'s}}.
}
Two key observations are in order. First, it follows by definition that $\alpha_0 + \alpha_1 + \cdots + \alpha_m = \binom{n}{2}$. Second, by definition of $\alpha_i$, it follows that
\aln{ 
\sum_{i=0}^{m} i \alpha_i = \sum_{e\in\binom{[n]}{2}} \sum_{j=1}^{m} G_j(e) = \sum_{e \in \binom{[n]}{2}} \sum_{j=1}^{m} G'_j(e)\label{eq: const}
}
is constant, independent of $\pistar_{12},\cdots,\pistar_{1m}$. It follows then that
\alns{ 
\eqref{eq: edge-prob} &= \pbr{1-p+p(1-s)^{m}}^{\alpha_0} \times \prod_{i=1}^{m} \pbr{ps^i (1-s)^{m-i}}^{\alpha_i} \\
&= \pbr{1-p+p(1-s)^m}^{\alpha_0} \times p^{\sum_{i=1}^m \alpha_i} \times \prod_{i=1}^m \pbr{s^i (1-s)^{m-i}}^{\alpha_i} \\
& = \pbr{1-p+p(1-s)^m}^{\alpha_0} \times p^{\binom{n}{2}-\alpha_0} \times \pbr{\frac{s}{1-s}}^{\sum_{i=1}^m i \alpha_i} \times (1-s)^{m \sum_{i=1}^m \alpha_i} \\
&=\pbr{\frac{1-p+p(1-s)^m}{p(1-s)^m}}^{\alpha_0} \times \pbr{p(1-s)^m}^{\binom{n}{2}} \times \pbr{\frac{s}{1-s}}^{\sum_{i=1}^{m} i \alpha_i} \\
&\propto \pbr{1+\frac{1-p}{p(1-s)^m}}^{\alpha_0},
}
where the last step uses~\eqref{eq: const}. Finally, since $\frac{1-p}{p(1-s)^m} > 0$, it follows that the log-likelihood satisfies
\alns{ 
\log\pbr{ \P{G_1,\cdots,G_m \ \vert \ \pistar_{12},\cdots,\pistar_{1m} }} \propto \text{const.} + \alpha_0,
}
i.e. maximizing the likelihood corresponds to selecting $\pi_{12},\cdots,\pi_{1m}$ to maximize $\alpha_0$ - the number of node pairs $e$ for which $G_1(e) = G_2(\pi_{12}(e)) = \cdots = G_m(\pi_{1m}(e)) = 0$. This is equivalent to minimizing the number of edges in the union graph $G_1 \vee G_2^{\pi_{12}}\vee \cdots \vee G_m^{\pi_{1m}}$, as desired.
\hfill $\Box{}$

\begin{remark}
    In the case of two graphs, minimizing the number of edges in the union graph $G_1 \vee_{\pi} G_2$ is equivalent to maximizing the number of edges in the intersection graph $G_1 \wedge_{\pi} G_2$. This is consistent with existing literature on two graphs~\citep{cullina2016improved, wu2022settling}.
\end{remark} 

\section{On negative results for exact recovery with two graphs} \label{apx: impossibility-two-graphs}

We build up to the proof of~\Cref{lem: impossibility-two-graphs} by presenting first a useful lemma about inhomogeneous random graphs.

\begin{lemma}\label{lem: second-moment-method}
    Let $G\sim\IHG(n,\bfQ)$ be an inhomogeneous random graph, and let $N \triangleq \binom{K}{2}$, where $K$ is the number of isolated nodes in $G$. Then,
    \aln{ 
    \P{N \geq 1} \geq \exp(-(2n-4)p_{\max}^2 - 6p_{\max}) \cdot \frac{\mu^2 - 2 \mu \max_{1\leq i \leq n}e^{-d_i}}{4(\mu^2 + \mu + 1)},
    }
    where $\mu \triangleq \sum_{i=1}^{n} e^{-d_i}$. 
\end{lemma}

\textit{Proof. }
    Let $Z_i$ denote the indicator random variable that node $i$ is isolated in $G$. Further, let $\tZ_{i} \sim \Bern(e^{-d_i})$ denote a collection of independent random variables and define $\tN \triangleq \sum_{1\leq i< j\leq n} \tZ_i\cdot \tZ_j$. We use the second moment method by relating the moments of $N$ and $\tN$.
    
    \vspace{0.1 in}
    \noindent \textit{Bounding the second moment.} Notice that
    \begin{align}
        \mathbb{E}[N^2] &= \sum_{1\leq i <j \leq n}\   \sum_{1\leq k < \ell \leq n} \mathbb{E}[Z_i Z_j Z_{k} Z_{\ell}]\,, \label{eq: sec-mom-N2}\\ \mathbb{E}[\tN^2] &= \sum_{1\leq i <j \leq n}\   \sum_{1\leq k < \ell \leq n} \mathbb{E}[\tZ_i \tZ_j \tZ_{k} \tZ_{\ell}]. \label{eq: sec-mom-tN2}
    \end{align}
     Consider the following three exhaustive cases. 
    \begin{itemize}
        \item[--] \textit{Case $1$: $i=k$ and $j=\ell$.}  Notice that
        \begin{align*}
            \mathbb{E}[\tZ_i\tZ_j\tZ_k\tZ_{\ell}] = \mathbb{E}[\tZ_i^2]\cdot \mathbb{E}[\tZ_j^2] = e^{-d_i-d_j}.
        \end{align*}
        Therefore, 
        \begin{align}
            \mathbb{E}[Z_i Z_j Z_k Z_{\ell}] \!=\!  (1-p_{ij}) \prod_{\substack{{(i,u)} \\ u \neq j}} (1-p_{iu}) \prod_{\substack{{(j,v)}\\v\neq i}} (1-p_{jv}) \stackrel{\text{(a)}}{\leq} e^{ - d_i - d_j + p_{ij}} 
            \stackrel{\text{(b)}}{\leq} e^{6p_{\max}} \cdot  \mathbb{E}[\tZ_{i} \tZ_{j} \tZ_{k} \tZ_{\ell}],  \label{eq: sec-mom-case1}
        \end{align}     
        where (a) uses that $1-x\leq e^{-x}$ for all $x$, and (b) uses that $p_{ij} \leq p_{\max} \leq 6p_{\max}$.
        
        \item[--] \textit{Case $2$: $j=k$.} Notice that
        \[
        \mathbb{E}[\tZ_i\tZ_j\tZ_k\tZ_{\ell}] = \mathbb{E}[\tZ_i] \cdot \mathbb{E}[\tZ_j^2] \cdot \mathbb{E}[\tZ_{\ell}] = e^{-d_i - d_j - d_{\ell}}.
        \]
        Therefore, similar to case $1$, we have
        \begin{align}
            \mathbb{E}[Z_iZ_jZ_kZ_{\ell}] &= \prod_{f \in \binom{\{i,j,\ell\}}{2}} (1-p_{f})\prod_{\substack{(i,u)\\u\neq j,\ell }}(1-p_{iu}) \prod_{\substack{(j,v)\\v\neq i,\ell }}(1-p_{jv}) \prod_{\substack{(\ell,w)\\w\neq i,j }}(1-p_{\ell w}) \nonumber \\
            & \leq e^{-d_i-d_j-d_{\ell}+p_{ij}+p_{i\ell}+p_{j\ell}} 
            \leq e^{6p_{\max}} \cdot \mathbb{E}[\tZ_{i} \tZ_{j} \tZ_{k} \tZ_{\ell}]  \label{eq: sec-mom-case2}
        \end{align}

        \item[--] \textit{Case $3$: $i \neq j \neq k \neq \ell $.} Notice that
        \begin{align*}
            \mathbb{E}[\tZ_i\tZ_j\tZ_k\tZ_{\ell}] = \mathbb{E}[\tZ_i] \cdot \mathbb{E}[\tZ_j] \cdot \mathbb{E}[\tZ_k] \cdot \mathbb{E}[\tZ_{\ell}] = e^{-d_i - d_j - d_k - d_{\ell}}.
        \end{align*}
        Therefore, similar to case $1$, we have
        \begin{align}
            \mathbb{E}[Z_iZ_jZ_kZ_{\ell}] &= \prod_{f \in \binom{\{i,j,k,\ell\}}{2}} \!\!\!(1-p_{f})\prod_{\substack{(i,u)\\u\neq j,k,\ell }}(1-p_{iu}) \prod_{\substack{(j,v)\\v\neq i,k,\ell }}(1-p_{jv})
            \prod_{\substack{(k,w)\\w\neq i,j,\ell }}(1-p_{kw}) \prod_{\substack{(\ell,x)\\x\neq i,j,k }}(1-p_{\ell x}) \nonumber \\
            & \leq e^{-d_i-d_j-d_k-d_{\ell}+p_{ij}+p_{ik}+p_{i\ell}+p_{jk}+p_{j\ell}+p_{k\ell}} 
            \leq e^{6p_{\max}} \cdot \mathbb{E}[\tZ_{i} \tZ_{j} \tZ_{k} \tZ_{\ell}] \label{eq: sec-mom-case3}
        \end{align}
    \end{itemize}
Consequently,
\begin{align}
    \mathbb{E}[\tN^2] &= \sum_{1\leq i < j\leq n} \sum_{1\leq k < \ell \leq n} \Big[ e^{-d_i-d_j}\indicator{i=k,j=\ell} + e^{-d_i-d_j-d_{\ell}}\indicator{j=k} + e^{-d_i-d_j-d_k-d_{\ell}}\indicator{i\neq j\neq k\neq \ell} \Big] \nonumber \\
    &\leq \sum_{1\leq i < j\leq n} \sum_{1\leq k < \ell \leq n} \Big[ e^{-d_i-d_j} + e^{-d_i-d_j-d_{\ell}} + e^{-d_i-d_j-d_k-d_{\ell}} \Big] 
    \leq \mu^2 + \mu^3 + \mu^4, \label{eq: form-tN2}
\end{align}
where $\mu \triangleq \sum_{1\leq i\leq n}e^{-d_i}$.
Combining equations~\eqref{eq: sec-mom-N2} and~\eqref{eq: sec-mom-tN2} with~\eqref{eq: sec-mom-case1},~\eqref{eq: sec-mom-case2},~\eqref{eq: sec-mom-case3} and~\eqref{eq: form-tN2} yields that
\begin{align}
    \mathbb{E}[N^2] \leq e^{6p_{\max}} \cdot \mathbb{E}[\tN^2] \leq e^{6p_{\max}} \mu^2 (\mu^2 + \mu + 1). \label{eq: N2-bound}
\end{align}

\noindent \textit{Bounding the first moment.} Notice that 
\begin{align}
    \mathbb{E}[\tN] = \sum_{1\leq i<j\leq n} \mathbb{E}[\tZ_i \tZ_j] = \sum_{1\leq i < j \leq n} e^{-d_i - d_j} = \frac{1}{2}\left( \mu^2 - \sum_{i=1}^{n} e^{-2d_i}\right) \geq \frac{\mu}{2}\pbr{\mu - \max_{1\leq i\leq n}e^{-d_i}}, \label{eq: form-tN1}
\end{align}
where the last step uses that $e^{-2d_k} \leq e^{-d_k} \cdot \max_{1\leq i \leq n} e^{-d_i}$ for each $k$.
On the other hand,
\begin{align}
    \mathbb{E}[N] &= \sum_{1\leq i<j\leq n} \E{Z_iZ_j} = \sum_{1\leq i<j\leq n}\Big[ (1-p_{ij}) \prod_{\substack{(i,u)\\ u\neq j}}(1-p_{iu}) \prod_{\substack{(j,v)\\ v\neq i}}(1-p_{jv}) \Big] \nonumber \\
    & \stackrel{\text{(a)}}{\geq} \sum_{1\leq i<j\leq n} \Big[ e^{-p_{ij}-p_{ij}^2} \prod_{\substack{(i,u)\\ u\neq j}} e^{-p_{iu}-p_{iu}^2} \prod_{\substack{(j,v)\\ v\neq i}} e^{-p_{jv}-p_{jv}^2}  \Big]  \nonumber \\
    & \stackrel{\text{(b)}}{\geq} e^{-(2n-1)p_{\max}^2} \cdot \sum_{1\leq i<j\leq n} e^{-d_i -d_j +p_{ij}} 
    \stackrel{\text{(c)}}{\geq} e^{-(2n-1)p_{\max}^2} \cdot \mathbb{E}[\tN], \label{eq: N1-bound}
\end{align}
where (a) uses that $1-x\geq e^{-x-x^2}$ for all $x\in [0,1/2]$, (b) uses that $e^{-p_{ij}^2} \geq e^{-p_{\max}^2}$, and (c) uses that $e^{p_{ij}} \geq 1$ for all $i$ and $j$. Finally, combining equations~\eqref{eq: N2-bound},~\eqref{eq: form-tN1} and~\eqref{eq: N1-bound}, we have that
\begin{align*}
    \frac{\mathbb{E}[N]^2}{\mathbb{E}[N^2]} \geq e^{- (4n-2)p_{\max}^2 - 6p_{\max}}  \frac{\mathbb{E}[\tN]^2}{\mathbb{E}[\tN^2]} \geq \exp\pbr{- (4n-2)p_{\max}^2 - 6p_{\max}} \frac{\mu^2 - 2\mu \max_{1\leq i \leq n} e^{-d_i}}{4(\mu^2 + \mu + 1)}, 
\end{align*}
where the last step uses that 
\[
\mathbb{E}[\tN]^2 \geq \frac{\mu^2}{4}\pbr{\mu^2  - 2 \mu \cdot \max_{1\leq i \leq n}e^{-d_i} + \max_{1\leq i \leq n} e^{-2d_i}} \geq \frac{\mu^2}{4}\pbr{ \mu^2 - 2 \mu \max_{1\leq i \leq n } e^{-d_i} }.
\]
Since $\P{N\geq 1} \geq \frac{\E{N}^2}{\E{N^2}}$, the proof is complete.
\hfill $\Box{}$

\subsection{Proof of Lemma~\ref{lem: impossibility-two-graphs}} \label{apxsub: impossibility-two-graphs}



\noindent \textit{Proof of~\Cref{lem: impossibility-two-graphs}.~~~~}Let $N$ be the number of pairs of isolated nodes in $H$. It follows from Lemma~\ref{lem: second-moment-method} that
    \aln{ 
    \P{N \geq 1}  &\geq \exp\pbr{-(2n-4)\newp_{\max}^2 - 6\newp_{\max}} \cdot \frac{\lambda^2 - 2 \lambda \max_{1\leq i \leq n}e^{-\newd_i}}{4(\lambda^2 + \lambda + 1)} \nonumber \\
    &\geq (1-o(1)) \cdot \frac{\lambda^2 - 2 \lambda \max_{1\leq i \leq n}e^{-\newd_i}}{4(\lambda^2 + \lambda + 1)}, \label{eq: ZXC}
    }
    where $\lambda \triangleq \sum_{i=1}^{n} e^{-\newd_i}$, and the last step in~\eqref{eq: ZXC} uses that $\newp_{\max} = o(1/\sqrt{n})$. Let $\newd_{[1]} \leq \newd_{[2]} \leq \cdots \leq \newd_{[n]}$ denote a sorted copy of the sequence $\newd_1,\cdots,\newd_n$. Consider the following three cases.
    \begin{itemize}
        \item[--] \textit{Case $1$. $\newd_{[1]} = \omega(1)$.} In this setting, it follows that
        \[ 
        \lambda \max_{1\leq i \leq n} e^{-\newd_{i}s(1-(1-s)^{m-1})} = \lambda e^{-\newd_{[1]}s(1-(1-s)^{m-1})} = o(1),
        \]
        since $\lambda = \Omega(1)$. Thus, 
        \[ 
        \P{N \geq 1} \geq (1-o(1)) \cdot \frac{\lambda^2 - o(1)}{4(\lambda^2 + \lambda + 1)} = \Omega(1).
        \]
        \item[--] \textit{Case $2$. $\newd_{[2]} = O(1)$.} In this setting, there exists a constant $C>0$ such that $C \leq e^{-\newd_{[2]}} \leq e^{-\newd_{[1]}}$. Let $A$ denote the adjacency matrix of $G$ and let $\newD_{[1]}$ and $\newD_{[2]}$ denote the degree of two nodes with mean degree $\newd_{[1]}$ and $\newd_{[2]}$ respectively. Then,
        \begin{align*}
            \P{\newD_{[1]} = 0, \newD_{[2]} = 0} &= \P{\newD_{[1]} = 0} \cdot \P{\newD_{[2]} = 0 \mid \newD_{[1]}=0}
            \geq \P{\newD_{[1]} = 0} \P{\newD_{[2]} = 0} \\
            & =  \prod_{i \neq [1]}(1-\newp_{[1]i}) \prod_{j\neq [2]}(1-\newp_{[2]j}) 
             \stackrel{\text{(a)}}{\geq}  \prod_{i \neq [1]} e^{-\newp_{[1]i} - \newp_{[1]i}^2} \prod_{j\neq [2]} e^{-\newp{[2]j}- \newp_{[2]j}^2} \\
            & = e^{-\newd_{[1]}} e^{-\newd_{[2]}} \cdot \exp\pbr{-\sum_{i \neq [1]} \newp_{[1]i}^2 - \sum_{j\neq [2]}\newp_{[2]j}^2} \\
            & \stackrel{\text{(b)}}{\geq} \exp\pbr{-2(n-1)\newp_{\max}^2} \cdot C^2  \stackrel{\text{(c)}}{=} \Omega(1).
        \end{align*}
        Here, (a) uses the inequality $1-x \geq e^{-x-x^2}$ whenever $x\in [0,1/2]$, (b) uses that $e^{-\newp_{ij}^2}  \geq e^{-\newp_{\max}^2}$ for all $i, j$ and also that $C \leq e^{-\newd_{[2]}} \leq e^{-\newd_{[1]}}$. Finally, (c) uses that $\exp\pbr{-2(n-1)\newp_{\max}^2} = \Omega(1)$ whenever $\newp_{\max}=o(1/\sqrt{n})$.

        \item[--] \textit{Case $3$. $\newd_{[2]} = \omega(1)$.} Consider the induced subgraph $G'$ of $G$ on the vertex set $V = V - \{[1]\}$. The mean degrees of $G'$ satisfy $\newd_{[i]}' \leq \newd_{[i]}$, and so it follows that 
        \[ 
        \lambda' \triangleq \sum_{i=2}^{n} e^{-d'_{[i]}} \leq \sum_{i=2}^{n} e^{-\newd_{[i]}} = \Omega(1),
        \]
        and the smallest mean degree in $G'$ diverges, i.e. $\newd_{[2]} = \omega(1)$. Equivalently, the graph $G'$ satisfies the conditions for Case $1$ above, and so there is a constant lower bound $C$ on the probability that the number of \textit{pairs} of isolated nodes $N'$ in $G'$ is at least one. If node $[1]$ does not connect to an isolated pair of nodes in $G'$, then the pair remains isolated in $G$ as well. Therefore,
        \begin{align*}
        \P{N \geq 1} \geq (1-\newp_{\max})^2\cdot \P{N' \geq 1} \geq (1-\newp_{\max})^2 \cdot C = \Omega(1).
        \end{align*}
        In all cases, we have that $\P{N\geq1}=\Omega(1)$, i.e. there is a non-vanishing probability that two isolated nodes exist in $H$. This concludes the proof.\hfill $\Box{}$
    \end{itemize}

\section[Stochastic domination of crossing edges]{Stochastic domination of crossing edges in $\widetilde{\mathcal{H}}(v)$: Proof of Lemma~\ref{lem: summed-stoch-dom}} \label{apx: stoch-dom}
The objective of this section is to build up to a proof of~\Cref{lem: summed-stoch-dom}. We start by making a simple observation about products of Binomial random variables.

\begin{lemma} \label{lem: conditional-stochastic-dominance}
Let $X_1,\cdots,X_m \sim \Bern(s)$ be i.i.d. random variables, and let $B = X_1 + \cdots + X_m$ denote their sum. For each $\ell$ in $\{1,2,\cdots,\lfloor m/2\rfloor\}$, define
\alns{ 
    T_{\ell} = \pbr{X_1 + \cdots + X_{\ell}}\pbr{X_{\ell +1} + \cdots + X_m}.
}
For any $\ell_1, \ell_2 \in \{1,2,\cdots,\lfloor m/2\rfloor\}$ such that $\ell_1 < \ell_2$, and for any $t \in \R$ and any $b \in \{0,1,\cdots,m\}$,
\aln{ 
    \P{ T_{{\ell}_1} > t \mid B = b} \leq \P{ T_{\ell_2} > t \mid B = b}. \label{eq: conditional-stochastic-dominance}
}
\end{lemma}

\noindent \textit{Proof of Lemma~\ref{lem: conditional-stochastic-dominance}. } Consider overlapping but exhaustive cases:

\vskip 0.1 in
\noindent \textit{Case $1$: $t < 0$.} Since $T_{\ell} \geq 0$ almost surely for all $\ell$, the inequality~\eqref{eq: conditional-stochastic-dominance} holds.

\vskip 0.1 in 
\noindent \textit{Case $2$: $t \geq b-1$.} Note that conditioned on $B=b$, it follows that $T_{1} \in \{0,b-1\}$. Therefore, the left hand side of~\eqref{eq: conditional-stochastic-dominance} equals zero, and the inequality holds.

\vskip 0.1 in 
\noindent \textit{Case $3$: $b = 0$ or $b=1$.} In this case, $T_\ell$ is identically zero for all $\ell$, so~\eqref{eq: conditional-stochastic-dominance} holds. 

\vskip 0.1 in 
\noindent \textit{Case $4$: $b \geq 2$ and $ 0\leq t< b-1$.} For any $\ell \in \{1,2,\cdots, \lfloor m/2\rfloor\}$, 
\aln{ 
\P{T_{\ell} > t \mid B=b} &= \frac{\P{\cbr{ \pbr{X_1+\cdots+X_{\ell}}\pbr{X_{\ell+1}+\cdots+X_m} > t}\cap\cbr{X_1+\cdots+X_m = b}}}{\P{X_1+\cdots+X_m = b}} \nonumber \\
&= \frac{\sum_{i: i(b-i)>t} \P{ \cbr{X_1+\cdots+X_{\ell} = i}\cap\cbr{X_{\ell+1}+\cdots+X_m = b-i}}}{\P{X_1+\cdots+X_m = b}} \nonumber \\
&\stackrel{\text{(a)}}{=} \frac{\sum_{i=1}^{b-1} \P{X_1+\cdots+X_{\ell} = i}\P{X_{\ell+1}+\cdots+X_m = b-i}}{\P{X_1+\cdots+X_m = b}} \nonumber \\
&\stackrel{\text{(b)}}{=}\frac{\sum_{i=1}^{b-1} \binom{\ell}{i} \binom{m-\ell}{b-i}}{\binom{m}{b}} \nonumber\\
&= \frac{\sum_{i=0}^{b} \binom{\ell}{i} \binom{m-\ell}{b-i} - \binom{m-\ell}{b} - \binom{\ell}{b}}{\binom{m}{b}} \nonumber \\
&= \frac{\binom{m}{b} - \binom{m-\ell}{b} - \binom{\ell}{b}}{\binom{m}{b}}, \label{eq: Prob-Tl}
}
where (a) used the fact that for any $t$ such that $0 \leq t < b-1$, it is true that $$\{i: i(b-i) > t\} = \cbr{1,2,\cdots,b-1}.$$
Here, the notation for binomial coefficients in (b) involves setting $\binom{n}{k} = 0$ whenever $k < 0$ or $k > n$. Let $f_{m,b}(\ell)$ denote the numerator of~\eqref{eq: Prob-Tl}, i.e. 
\alns{ 
f_{m,b}(\ell) := \binom{m}{b} - \binom{m-\ell}{b} - \binom{\ell}{b}
}
It suffices to show that $f_{m,b}(\ell) - f_{m,b}(\ell - 1) \geq 0$ for all $\ell \in \cbr{2,\cdots,\lfloor m/2\rfloor}$. Indeed, 
\alns{ 
f_{m,b}(\ell) - f_{m,b}(\ell - 1) &= \binom{m-\ell + 1}{b} - \binom{m-\ell}{b} - \pbr{ \binom{\ell}{b}-\binom{\ell-1}{b}} \\
& \stackrel{\text{(c)}}{=} \binom{m-\ell}{b-1} - \binom{\ell - 1}{b-1} \geq 0,
}
whenever $m-\ell \geq \ell - 1$, i.e. $\ell \leq \lfloor m/2\rfloor$. Here, (c) uses the identity $\binom{n}{k} = \binom{n-1}{k-1} + \binom{n-1}{k}$, and the fact that $\binom{n_1}{k} \geq \binom{n_2}{k}$ whenever $n_1\geq n_2$. This concludes the proof.
\hfill $\Box$

\begin{corollary} \label{cor: stoch-dom}
Let $F$ be a collection of edges in the parent graph $G$. For any edge $e_r \in F$, let $X^r_i$ denote the indicator random variable $G_i'(e_r) \sim \Bern(ps)$. For each $\ell$ in $\{1,\cdots,\lfloor m/2\rfloor\}$, define
\alns{ 
T_{\ell}^r = (X_1^r + \cdots + X_{\ell}^r)(X_{\ell+1}^r + \cdots + X_m^r).
}
Then, for any $\ell_1,\ell_2 \in \{1,\cdots,\lfloor m/2\rfloor\}$ such that $\ell_1 < \ell_2$, the following stochastic ordering holds
\alns{ 
\sum_{r = 1}^{|F|} T^r_{\ell_1} \preceq \sum_{r = 1}^{|F|} T^r_{\ell_2}. 
}
\end{corollary}
\textit{Proof. }
It suffices to show that $T^r_{\ell_1} \preceq T^r_{\ell_2}$ for each $r$, since the edges are independent. For any $t$,
\alns{ 
\P{T^r_{\ell_1} > t} = \sum_{b=0}^{m} \P{B = b} \P{ T_{\ell_1}^r > t | B = b} \leq \sum_{b=0}^{m} \P{B = b} \P{ T_{\ell_2}^r > t | B = b} = \P{T^r_{\ell_2} > t},
}
which concludes the proof.
\hfill $\Box{}$

With this, we are ready to prove~\Cref{lem: summed-stoch-dom}. The lemma is restated for convenience.

\thmstochdom*
\textit{Proof. }
    Let $\ell_1, \ell_2 \in \{1,\cdots,\lfloor m/2\rfloor \}$ such that $\ell_1 < \ell_2$. Let $t \in \R$. Consider the parent graph $G$ and label the set of incident edges on $v$ as $\{e_1,\cdots,e_{\deg_G(v)}\}$. Denote by $X_i^r$ the indicator random variable $G_i'(e_r) \sim \Bern(ps)$. It follows that
    \alns{ 
    \P{\ct_v(U_{\ell_2}) > t} &= \P{ \sum_{i=1}^{\ell_2} \sum_{j=\ell_2 + 1}^m \deg_{G_i'\wedge G_j'}(v) \geq t} \\
    &= \P{\sum_{i=1}^{\ell_2} \sum_{j=\ell_2 +1}^m \sum_{r=1}^{\deg_G(v)} X^r_i X^r_j > t} \\
    &= \sum_{d=0}^n \P{\deg_G(v) = d} \P{\sum_{r=1}^d \pbr{ (X_1^r + \cdots + X_{\ell_2}^r)(X_{\ell_2+1}^r + \cdots + X_m^r)} > t} \\
    & \stackrel{\text{(a)}}{\geq} \sum_{d=0}^n \P{\deg_G(v) = d} \P{\sum_{r=1}^d \pbr{ (X_1^r + \cdots + X_{\ell_1}^r)(X_{\ell_1+1}^r + \cdots + X_m^r)} > t} \\
    & = \P{\sum_{i=1}^{\ell_1} \sum_{j=\ell_1 +1}^m \sum_{r=1}^{\deg_G(v)} X^r_i X^r_j > t} \\
    & = \P{ \sum_{i=1}^{\ell_1} \sum_{j=\ell_1 + 1}^m \deg_{G_i'\wedge G_j'}(v) \geq t}  \\
    & = \P{\ct_v(U_{\ell_1}) > t},
    }
    as desired. Here, (a) uses~\Cref{cor: stoch-dom}.
\hfill $\Box{}$

\end{document}